\documentclass[twocolumn,prl,longbibliography,superscriptaddress]{revtex4-2}

\usepackage[latin1]{inputenc}
\usepackage{graphicx}
\usepackage{amsmath}
\usepackage{amsthm}
\usepackage{bm}
\usepackage{layout}
\usepackage{float}
\usepackage{amsfonts}
\usepackage{amssymb}%
\usepackage[margin=0.75in]{geometry}
\usepackage{color}
\usepackage{soul}
\usepackage{mathtools}
\usepackage[colorlinks=true,citecolor=blue]{hyperref}
\theoremstyle{definition}
\newtheorem{theorem}{Theorem}

\usepackage[capitalise,compress]{cleveref}

\newcommand{\oket}[1]{\left | #1 \right)}

\renewcommand{\epsilon}{\varepsilon}
\renewcommand{\O}[1]{ \mathcal{O}(#1)}

\newcommand{\norm}[1]{\left\|#1\right\|}

\newcommand{\comm}[1]{\left[#1\right]}

\newcounter{para}
\newcommand{\dist}{\textrm{dist}}

\makeatletter
\newcommand*\bigcdot{\mathpalette\bigcdot@{.5}}
\newcommand*\bigcdot@[2]{\mathbin{\vcenter{\hbox{\scalebox{#2}{$\m@th#1\bullet$}}}}}

\makeatother

\usepackage{array}
\newcolumntype{L}{>{$}l<{$}} 
\newcolumntype{C}{>{$}c<{$}} 
\newcolumntype{R}{>{$}r<{$}} 

\usepackage{xcolor}

\usepackage{mathtools}

\usepackage{xr}
\makeatletter
\newcommand*{\addFileDependency}[1]{
  \typeout{(#1)}
  \@addtofilelist{#1}
  \IfFileExists{#1}{}{\typeout{No file #1.}}
}
\makeatother

\usepackage{tikz}
\renewcommand{\L}{\mathcal L}
\renewcommand{\P}{\mathbb P}

\usepackage{mdframed}
\newmdenv[topline=false,rightline=false,bottomline=false,linewidth=2pt,linecolor=white!60!black,]{leftborder}

\usepackage{qcircuit}

\makeatletter
\newcommand{\subalign}[1]{%
  \vcenter{%
    \Let@ \restore@math@cr \default@tag
    \baselineskip\fontdimen10 \scriptfont\tw@
    \advance\baselineskip\fontdimen12 \scriptfont\tw@
    \lineskip\thr@@\fontdimen8 \scriptfont\thr@@
    \lineskiplimit\lineskip
    \ialign{\hfil$\m@th\scriptstyle##$&$\m@th\scriptstyle{}##$\hfil\crcr
      #1\crcr
    }%
  }%
}
\makeatother

\renewcommand{\section}[1]{\textit{#1.---}}
\begin{document}
\title{The Lieb-Robinson light cone for power-law interactions}
\date{\today}

\author{Minh C. Tran}
\email{minhtran@umd.edu} 

\affiliation{Joint Center for Quantum Information and Computer Science,
NIST/University of Maryland, College Park, MD 20742, USA}
\affiliation{Joint Quantum Institute, NIST/University of Maryland, College Park, MD 20742, USA}

\author{Andrew Y. Guo}
\affiliation{Joint Center for Quantum Information and Computer Science,
NIST/University of Maryland, College Park, MD 20742, USA}
\affiliation{Joint Quantum Institute, NIST/University of Maryland, College Park, MD 20742, USA}

\author{Christopher L. Baldwin}
\affiliation{Joint Center for Quantum Information and Computer Science,
NIST/University of Maryland, College Park, MD 20742, USA}
\affiliation{Joint Quantum Institute, NIST/University of Maryland, College Park, MD 20742, USA}

\author{Adam Ehrenberg}
\affiliation{Joint Center for Quantum Information and Computer Science,
NIST/University of Maryland, College Park, MD 20742, USA}
\affiliation{Joint Quantum Institute, NIST/University of Maryland, College Park, MD 20742, USA}

\author{Alexey V. Gorshkov}
\affiliation{Joint Center for Quantum Information and Computer Science,
NIST/University of Maryland, College Park, MD 20742, USA}
\affiliation{Joint Quantum Institute, NIST/University of Maryland, College Park, MD 20742, USA}

\author{Andrew Lucas}
\affiliation{Department of Physics, University of Colorado, Boulder CO 80309, USA}
\affiliation{Center for Theory of Quantum Matter, University of Colorado, Boulder CO 80309, USA}

\begin{abstract}
	The Lieb-Robinson theorem states that information propagates with a finite velocity in quantum systems on a lattice with nearest-neighbor interactions.  What are the speed limits on information propagation in quantum systems with power-law interactions, which decay as $1/r^\alpha$ at distance $r$? Here, we present a definitive answer to this question for all exponents $\alpha>2d$ and all spatial dimensions $d$.  Schematically, information takes time at least $r^{\min\{1, \alpha-2d\}}$ to propagate a distance~$r$. As recent state transfer protocols saturate this bound, our work closes a decades-long hunt for optimal Lieb-Robinson bounds on quantum information dynamics with power-law interactions.
\end{abstract}
\maketitle

Over a century ago, Einstein realized that there is a speed limit to information propagation. If no physical object or signal can travel faster than light, then the speed of light itself must constrain the dynamics of quantum information and entanglement.  In ordinary quantum systems, however, \emph{emergent} speed limits can arise that place more stringent restrictions on information propagation than does the speed of light.  For example, in quantum spin systems with nearest-neighbor interactions on a lattice, Lieb and Robinson proved in 1972 that there is a \emph{finite velocity} of information propagation \cite{LR}.

Of course, most non-relativistic physical systems realized in experiments include long-range interactions such as the Coulomb interaction, the dipole-dipole interaction, or the van-der-Waals interaction.  Each of these decays with distance as a power law $1/r^\alpha$ for some exponent $\alpha$.  What is the fundamental speed limit on the propagation of quantum information in these systems?

Despite the importance of this question in designing and constraining the operation of future quantum technologies \cite{PhysRevA.46.R6797,fossfeig2016entanglement,linkeExperimentalComparisonTwo2017,Deshpande2018,landsmanVerifiedQuantumInformation2019}, bounding information propagation in systems with power-law interactions has been a notoriously challenging mathematical physics problem.  In 2005, Hastings and Koma~\cite{HK} showed that it takes a time $t\gtrsim\log r$ to send information a distance $r$, for all $\alpha>d$, where $d$ is the dimension of the lattice.  By analogy to Einstein's relativity, we say that there is at least a ``logarithmic light cone" for such power-law interactions.  However, it was suspected that this bound was far from tight, and ten years later it was shown that $t \gtrsim r^\gamma $, for an exponent $0<\gamma<1$ when $\alpha>2d$ \cite{Foss-FeigG,elseImprovedLiebRobinsonBound2018,tranLocalityDigitalQuantum2019a}.
In 2019, Chen and Lucas~\cite{chenFiniteSpeedQuantum2019} proved the existence of a linear light cone ($t\gtrsim r$) for all $\alpha > 3$ in $d = 1$; Kuwahara and Saito~\cite{kuwaharaStrictlyLinearLight2020} later generalized this result to higher dimensions, finding a linear light cone for all $\alpha > 2d+1$.   These recent results prove that power-law interactions are, for all practical purposes, entirely local for sufficiently large~$\alpha$.

A natural question is then how small $\alpha$ must be in order to \emph{break} a linear light cone. Fast state-transfer and entanglement-generation protocols developed in the past year~\cite{Zachary17,tranHierarchyLinearLight2020a,kuwaharaStrictlyLinearLight2020,2020arXiv201002930T} have ultimately demonstrated that the time $t$ required to send information a distance $r$ obeys $t\lesssim r^{\min(\alpha-2d,1)}$ for any $\alpha>2d$ and $t\lesssim r^{o(1)}$ for $\alpha<d$, where $o(1)$ is an arbitrarily small constant.   Combining all best known results in the literature leads to the diagram shown in Fig. \ref{fig:light-cones}, which compares known information-transfer protocols to corresponding Lieb-Robinson bounds.

\begin{figure}[t]
	\includegraphics[width = 0.45\textwidth]{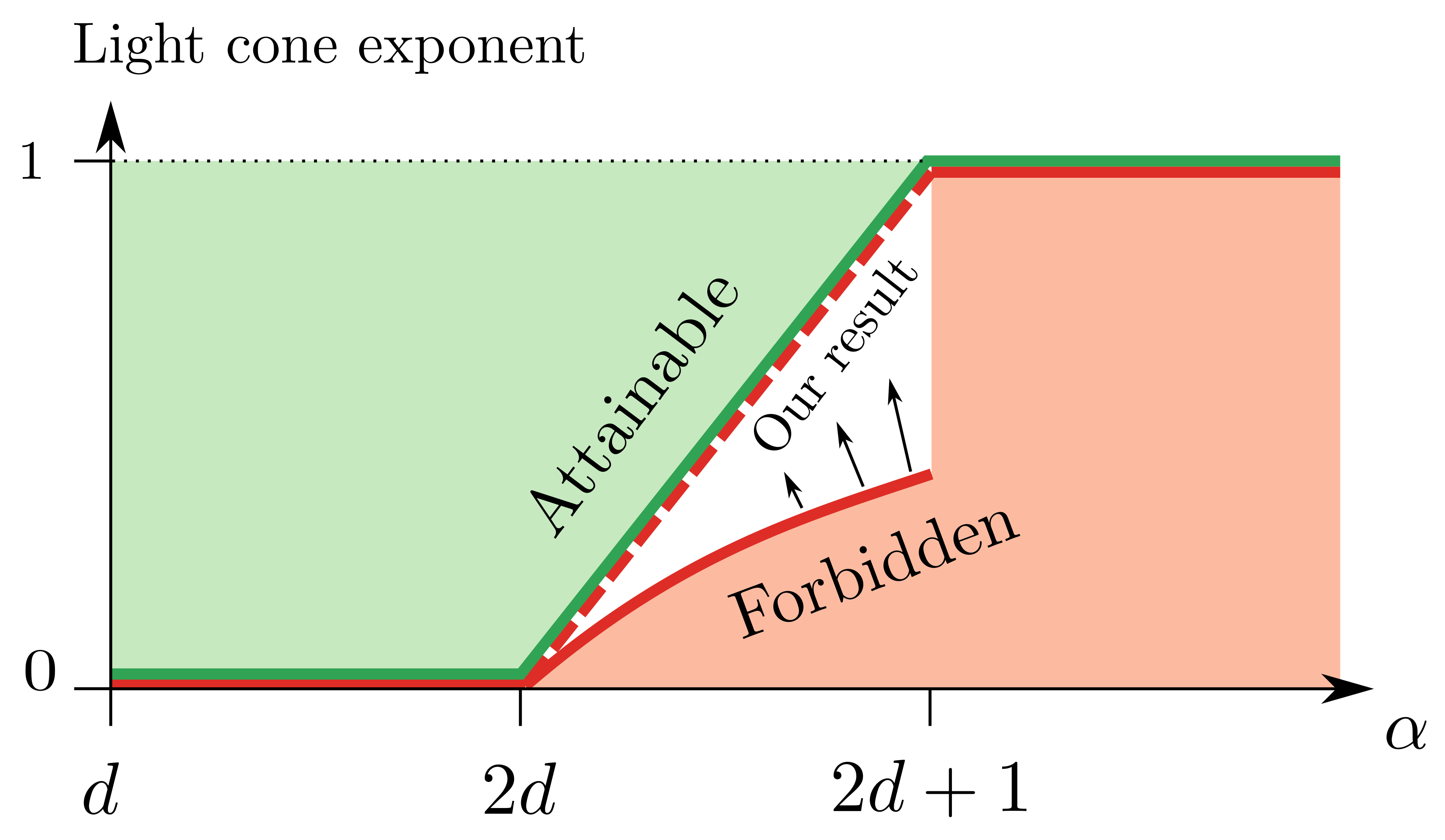}
	\caption{The gap in the Lieb-Robinson literature in $d>1$ dimensions.
	The red solid lines represent the exponent $\gamma$ of the Lieb-Robinson light cone $t \gtrsim r^{\gamma}$ in literature.
	The green solid lines correspond to the light cone exponents of best-known information-propagating protocols.
	Accordingly, the green region corresponds to attainable light cone exponents, whereas the red region is forbidden by the known bounds.
	Our result (red dashed line) closes the gap in our understanding of the Lieb-Robinson light cone.
	}
	\label{fig:light-cones}
\end{figure}

In this Letter, we complete this extensive literature on Lieb-Robinson bounds for power-law interactions~\cite{HK,Foss-FeigG,elseImprovedLiebRobinsonBound2018,tranLocalityDigitalQuantum2019a,chenFiniteSpeedQuantum2019,Zachary17,kuwaharaStrictlyLinearLight2020,tranHierarchyLinearLight2020a,2020arXiv201002930T,NachtergaeleOS2006,Nachtergaele2006,GongFF,Storch15,NRSS09,SHKM10,SH10}, by proving that quantum information is contained within the Lieb-Robinson light cone  $t\gtrsim r^{\min(\alpha-2d-\epsilon,1)}$, for any $\epsilon>0$.  This result closes the remaining gap between bounds and protocols in Fig. \ref{fig:light-cones}, and  concludes the fifteen-year quest to understand the fundamental speed limit on quantum information in the presence of power-law interactions.
We sketch the proof of the result in the main text and refer readers to the Supplemental Material (SM)~\cite{SM} for a rigorous treatment.

\section{Main result}We consider a $d$-dimensional regular lattice $\Lambda$, a finite-level system at every site of the lattice, and a two-body power-law Hamiltonian $H(t)$ with an exponent $\alpha$ supported on the lattice.
Specifically, we assume
$H(t) = \sum_{i,j\in \Lambda} h_{ij}(t)$ is a sum of two-body terms $h_{ij}$ supported on sites $i,j$ such that $\norm{h_{ij}(t)} \leq 1/\dist(i,j)^\alpha$ for all $i\neq j$,
where $\norm{\cdot}$ is the operator norm and
$\dist(i,j)$ is the distance between $i,j$.
In the following discussion, we assume $\Lambda$ is a hypercubic lattice of qubits for simplicity.

We use $\L$ to denote the Liouvillian corresponding to the Heisenberg evolution under Hamiltonian $H$, i.e. $\L\oket{O} \equiv \oket{i [H,O]}$ for any operator $O$, and use $e^{\L t} \oket{O} \equiv \oket{O(t)}$ to denote the time-evolved version of the operator $O$.
We also use $\P_{r}^{(i)}\oket{O}$ to denote an operator constructed from $O$ by decomposing $O$ into a sum of Pauli strings and removing strings that are supported entirely within a ball of radius $r$ from $i$.
Colloquially speaking, $\P_{r}^{(i)}\oket{O}$ is the component of 
 $O$ that has non-trivial support on sites a distance at least $r$ from site $i$.
If $i$ is the origin of the lattice, we drop the superscript $i$ and simply write $\P_r$ for brevity.

Given a unit-norm operator $O$ initially supported at the origin, our main result is a bound on how much $O$ spreads to a distance $r$ and beyond under the evolution $e^{\L t}$:
\begin{theorem}\label{thm:main-bound}
	For any $\alpha\in (2d,2d+1)$ and an arbitrarily small $\epsilon >0$, there exist constants $c,C\geq 0$ such that
	\begin{align}
		\norm{\P_r e^{\L t}\oket{O}} \leq
		C \left(\frac{t}{r^{\alpha-2d-\epsilon}}\right)^{\frac{\alpha-d}{\alpha-2d}-\frac{\epsilon}{2}}\label{eq:main-bound}
	\end{align}
	holds for all $1\leq t \leq c r^{\alpha-2d-\epsilon}$.
\end{theorem}
Because $\norm{\P_r e^{\L t}\oket{O}}$ can be both upper- and lower-bounded by linear functions of $\sup_{A}\norm{\comm{A,e^{\L t}O}}$, where $A$ is a unit-norm operator supported at least a distance $r$ from $O$,
\cref{eq:main-bound} is equivalent to a bound on the unequal-time commutators commonly used in the Lieb-Robinson literature.

For $\alpha \in (2d,2d+1)$, by setting the left-hand side of \cref{eq:main-bound} to a constant, \Cref{thm:main-bound} implies the light cone $t \gtrsim r^{\alpha-2d-\epsilon}$ for some $\epsilon$ that can be made arbitrarily small.
Note that our definition does not require $\norm{h_{ij}}$ to decay exactly as $1/\dist(i,j)^\alpha$; it may actually decay faster than $1/\dist(i,j)^\alpha$ and still satisfy the condition of a power-law interaction with an exponent $\alpha$.
Therefore, for $\alpha\geq 2d+1$ and power-law Hamiltonians $H = \sum_{ij} h_{ij}$ satisfying $\norm{h_{ij}}\leq 1/\dist(i,j)^\alpha < 1/\dist(i,j)^{2d+1-\epsilon}$, \cref{thm:main-bound} implies a linear light cone $t \gtrsim r^{1-2\epsilon}$.

\section{Sketch of proof}
For simplicity, we assume here that the lattice diameter is $\O{r}$.
We show in the SM~\cite{SM} that interactions whose ranges are much larger than $r$ do not contribute significantly to the dynamics of $O$ and, therefore, can be safely removed from the Hamiltonian.
Our strategy is to group the interactions of the Hamiltonian by their ranges, prove a bound for short-range interactions, and recursively add longer-range interactions to the Hamiltonian.

Specifically, we choose $\ell_k \equiv L^k$ for $k=1,\dots,n$, where $L,n$ are parameters to be chosen later.
We use $H_k$ to denote those terms of $H$ with range at most $\ell_k$
and use $\L_k \equiv i[H_k,\cdot]$ to denote the corresponding Liouvillian.
We start with the standard Lieb-Robinson bound for $H_1$~\cite{LR}:
\begin{align}
	\norm{\P_r e^{\L_1 t}\oket{O}} \lesssim e^{\frac{v_1 t - r}{\ell_1}},\label{eq:bound-for-H1}
\end{align}
where $v_1 \propto \ell_1 = L$ is the rescaled Lieb-Robinson velocity, and prove a bound for $H_2$ by adding $V_2 \equiv H_2 - H_1$, i.e., interactions of range between $\ell_1$ and $\ell_2$, to the Hamiltonian $H_1$.

For that, we move into the interaction picture of $H_1$ so that we can decompose the evolution $e^{\L_2 t} = e^{\L_{2,I} t} e^{\L_1 t}$ into two consecutive evolutions, where $e^{\L_{2,I}t}$ is the evolution under $V_{2,I} \equiv e^{\L_1 t} V_2$.
Loosely speaking, the light cone induced by $H_2$ will be a ``sum'' of the light cones induced by $H_1$ and $V_{2,I}$ individually (see the SM~\cite{SM} for a proof.)
With the light cone of $H_1$ given by \cref{eq:bound-for-H1}, our task is to find the light cone of $V_{2,I}$.

For this purpose, we consider the structure of $V_{2,I}$ and show that, with a suitable rescaling of the lattice, the interactions in $V_{2,I}$ decay exponentially with distance.
We then obtain the light cone of $V_{2,I}$ using the standard Lieb-Robinson bound on the rescaled lattice.
Specifically, we divide the lattice into non-overlapping hypercubes of length $\ell_2$ (see \cref{fig:supersite}).
Given $x,y$ as the centers of two hypercubes, we define $\dist(x,y)/\ell_2$ to be the rescaled distance between the hypercubes.
We shall estimate the strength of the interaction between hypercubes under the Hamiltonian $V_{2,I}$.

\begin{figure}[t]
	\includegraphics[width = 0.40\textwidth]{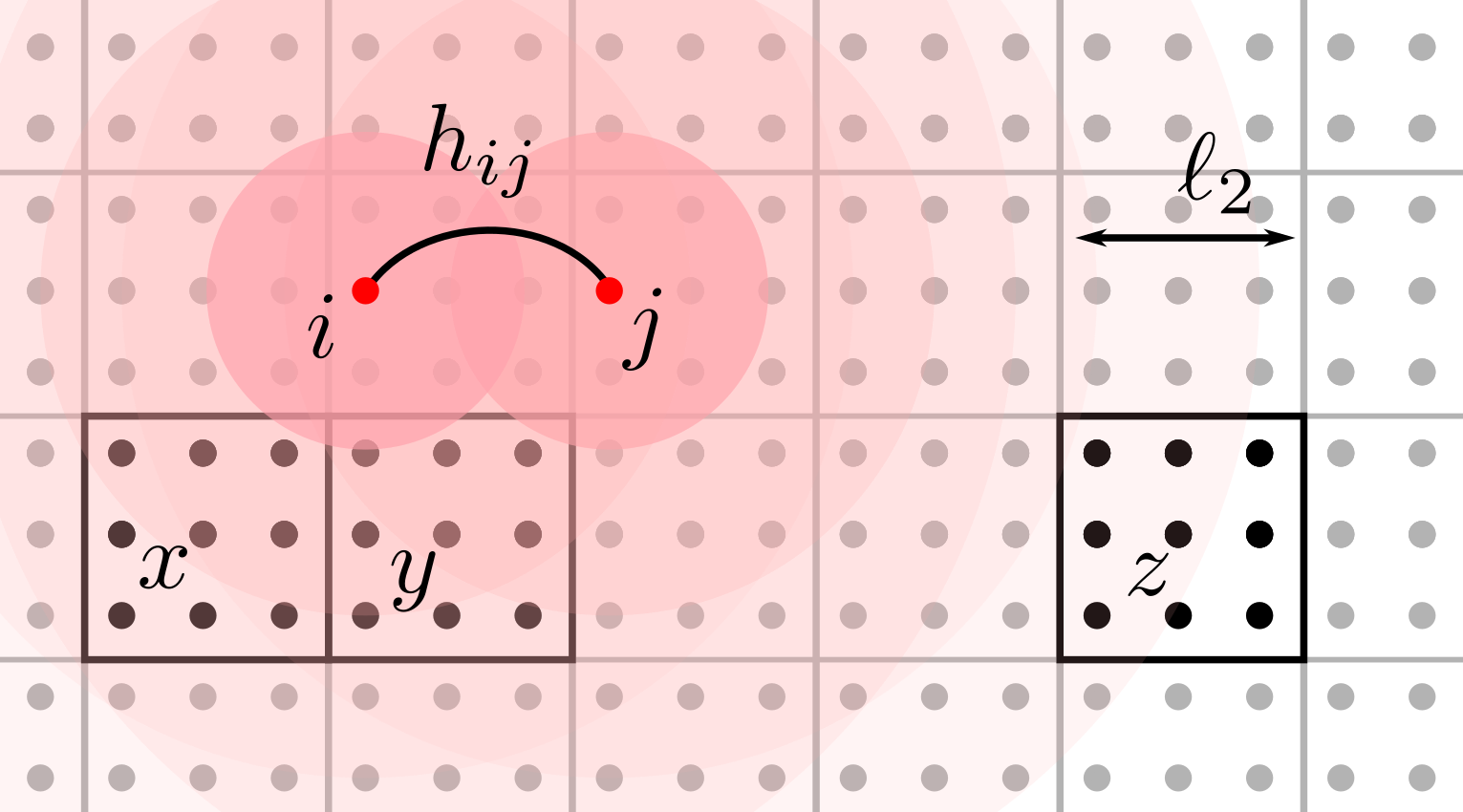}
	\caption{We study the structure of $V_{2,I}$ by dividing the lattice into hypercubes of length $\ell_2$ (labeled by $x,y$, and $z$ for example).
	In the interaction picture, how much each $e^{\L_1 t} h_{ij}$ contributes to the pair-wise ``effective interaction'' between two hypercubes depends on how strongly the support of $e^{\L_1 t} h_{ij}$ (represented by the shaded area) overlaps with the hypercubes.
	Because of the bound in \cref{eq:bound-for-H1}, the evolved operator $e^{\L_1 t} h_{ij}$ is largely confined to the light cones induced by $\L_1$ around $i$ and $j$ (the smallest disks around $i$ and $j$).
	The component of $e^{\L_1 t} h_{ij}$ supported outside this light cone is exponentially suppressed with distance (represented by lighter shade).
	Consequently, the effective interaction between the hypercubes $x$ and $z$ is exponentially smaller than the one between $x$ and $y$.
	}
	\label{fig:supersite}
\end{figure}

We first consider the case $t = 0$ so that $V_{2,I} = V_2$.
Because each interaction in $V_{2}$ has range at most $\ell_2$, no interaction $h_{ij}$ is supported on two distinct hypercubes unless they are nearest neighbors.
Therefore, only nearest-neighboring hypercubes may interact under $V_{2,I} = V_2$.

The case $t>0$ is slightly more complicated.
The support of an interaction $h_{ij}$ in $V_2$ may expand under $e^{\L_1 t}$, and, hence,
non-nearest-neighboring hypercubes may interact with each other.
However, due to \cref{eq:bound-for-H1}, the support of $e^{\L_1 t} h_{ij}$ would largely remain  inside the 
balls of radius $v_1 t$ around $i,j$.
The interactions between hypercubes are exponentially suppressed with distance by \cref{eq:bound-for-H1}.
Therefore, the system of hypercubes would interact via a nearly finite-range interaction (see \cref{fig:supersite}).

To apply the standard Lieb-Robinson bound for this system of hypercubes, we estimate the maximum effective interaction between any pair of nearest-neighboring hypercubes centered on $x,y$.
In particular, assuming $v_1 t \leq \ell_2$, the primary contributions to such an interaction come from $\propto \ell_2^d \times \ell_2^d = \ell_2^{2d}$ interaction terms $e^{\L_1 t}h_{ij}$ whose light cones under $H_1$ may overlap with the hypercubes.
Because each interaction $h_{ij}$ has norm at most $1/\ell_1^\alpha$ by our assumption, the total contribution to the interactions between the cubes $x,y$ is $\O{{\ell_2^{2d}}/{\ell_1^\alpha}}$.
Applying the standard finite-range Lieb-Robinson bound on the system of hypercubes, where the maximum energy per interaction is $\O{{\ell_2^{2d}}/{\ell_1^\alpha}}$ and the distance is rescaled by a factor $\ell_2$, we obtain the bound for the evolution under $V_{2,I}$:
\begin{align}
	\norm{\P_r e^{\L_I t} \oket{O}} \lesssim \exp\left(\mathcal{O}\left(\frac{\ell_2^{2d}}{\ell_1^\alpha}\right)t - \frac{r}{\ell_2} \right)
    \equiv e^{\frac{\Delta v t - r}{\ell_2}}, \label{eq:bound-for-V2}
\end{align}
where $\Delta v = \O{\ell_2^{2d+1}/\ell_1^\alpha}$.

After getting the light cone for the evolution under $V_{2,I}$, we now combine it with the evolution under $H_1$ to obtain the light cone of $H_2$.
Intuitively, the evolutions under $H_1$ and $V_{2,I}$ for time $t$ may each grow the support radius of an operator by $v_1 t$ and $\Delta v t$ respectively.
Therefore, one would expect an operator evolved under $H_1$ and $V_{2,I}$ consecutively, each for time $t$, may have the support radius at most $(v_1 + \Delta v )t$.
In the SM~\cite{SM}, we show that
\begin{align}
	\norm{\P_r e^{\L_2 t}\oket{O}}
	=\norm{\P_r e^{\L_{2,I} t} e^{\L_1 t} \oket{O}} \lesssim e^{\frac{v_2 t - r}{\ell_2}}, \label{eq:bound-for-H2}
\end{align}
where
\begin{align}
	v_2 \propto \log(r) v_1 + \Delta v = \log(r) v_1 + \frac{\ell_2^{2d+1}}{\ell_1^\alpha}
\end{align}
and $r$ is the diameter of the lattice $\Lambda$.
The additional factor of $\log(r)$ (compared to our intuition) comes from the enhancement to the operator spreading due to the increased support size after the first evolution $e^{\L_1 t}$.

Up to this point, we have used the bound \cref{eq:bound-for-H1} for $H_1$ to prove a bound for $H_2$ [\cref{eq:bound-for-H2}], which has the same form.
Repeating this process, we arrive at similar bounds for $H_k$ $(k = 3,4,\dots,n)$:
\begin{align}
	\norm{\P_r e^{\L_k t}\oket{O}}
	\lesssim e^{\frac{v_k t - r}{\ell_k}}, \label{eq:bound-for-Hk}
\end{align}
where the velocity $v_k$ is defined iteratively:
\begin{align}
	v_k \propto \log (r) v_{k-1} +\frac{\ell_k^{2d+1}}{\ell_{k-1}^\alpha}
	.\label{eq:vkiterative}
\end{align}
 Increasing $k$ makes the bound in \cref{eq:bound-for-Hk} applicable for longer and longer interactions.
However, doing so also increases $\ell_k$, resulting in weaker and weaker bounds.
In particular, if $\ell_k >r$, \cref{eq:bound-for-Hk} becomes trivial even for $t\leq r/v_k$.
 Therefore, we stop the iteration at $k = n$ such that $\ell_n$ is slightly smaller than $r$.
 Specifically, we choose $n$ such that $\ell_n = L^n = r/ \chi(t,r)$, where $\chi(t,r)>1$ is a function of $t,r$.
 For $v_n t\leq r/2$, the right-hand side of \cref{eq:bound-for-Hk} becomes
 \begin{align}
 	e^{\frac{v_n t - r}{\ell_n}} \lesssim  e^{-\frac{r}{2\ell_n}}
 	\lesssim e^{-\frac{1}{2}\chi(t,r)}
 	\lesssim \frac{1}{\chi(t,r)^{\omega}},
 \end{align}
 where we upper-bound an exponentially decaying function of $\chi(t,r)$ by a power-law decaying function of $\chi(t,r)$ with an exponent $\omega>0$.
 Choosing $\chi(t,r) = (r^{\alpha-2d}/t)^{\zeta}$, where $\zeta>0$ is an arbitrarily small constant, and $\omega = \frac{\alpha-d}{\zeta(\alpha-2d)}$, we obtain the desired bound
\begin{align}
	\norm{\P_r e^{\L_n t}\oket{O}}
	\lesssim \left(\frac{t}{r^{\alpha-2d}}\right)^{\frac{\alpha-d}{\alpha-2d}}. \label{eq:bound-for-Hn}
\end{align}

Note that \cref{eq:bound-for-Hn} only holds for $t \leq r/2v_n$.
To maximize the range of validity of \cref{eq:bound-for-Hn}, we aim to choose $L$ such that $v_n$ is as small as possible.
Without the second term in \cref{eq:vkiterative}, we would expect $v_k$ to increase by a factor of $\log r$ between iterations.
Meanwhile, given $\ell_k = L^k$, the second term in \cref{eq:vkiterative} also increases by a factor $L^{2d+1-\alpha}$ in every iteration.
Choosing $L^{2d+1-\alpha} \propto \log r$ so that the two terms in \cref{eq:vkiterative} have roughly equal contributions to $v_k$, we expect
 \begin{align}
     v_n \propto (\log r)^n \propto L^{n(2d+1-\alpha)}
     = \left(\frac{r}{\chi(t,r)}\right)^{2d+1-\alpha}
 \end{align}
up to a small logarithmic correction in $r$.
Substituting the earlier choice of $\chi(t,r)$, we have
\begin{align}
 	v_n t \propto r \left(\frac{t}{r^{\alpha-2d}}\right)^{1+o(1)} \leq r,
 \end{align}
 where $o(1)$ represents an arbitrarily small constant,
 for all $t \leq r^{\alpha-2d}$.
 Therefore, the bound in \cref{eq:bound-for-Hn} holds as long as $t \lesssim r^{\alpha-2d}$.

The bound in \cref{eq:bound-for-Hn} applies to the Hamiltonian $H_n$ constructed from $H$ by taking interactions of range at most $\ell_n$, which is slightly smaller than $r$ for all $t\lesssim r^{\alpha-2d}$.
To add interactions of range larger than $\ell_n$ to the bound, we use the identity~\cite{chenOperatorGrowthBounds2019}:
\begin{align}
	e^{\L t} = e^{\L_n t} + \sum_{\substack{i,j:\dist({i,j})>\ell_n}}
	\int_0^t ds\ e^{\L (t-s)} \L_{h_{ij}} e^{\L_n s} , \label{eq:add-long-range}
\end{align}
where $\L_{h_{ij}} = i[h_{ij},\cdot]$ is the Liouvillian corresponding to the interaction $h_{ij}$.
We will argue that the contribution from the second term of the right-hand side to the bound on $\norm{\P_r e^{\L t}\oket{O}}$ is small.

Note that $\L_{h_{ij}} e^{\L_n s} \oket{O}$ vanishes if $e^{\L_n s} \oket{O}$ has no support on the sites $i,j$.
Suppose site $i$ is closer to the origin than site $j$.
Then, most contributions to the right-hand side of \cref{eq:add-long-range}  come from terms $h_{ij}$ where $i$ lies within the light cone of $e^{\L_n s}\oket{O}$.
Let $\mathcal V$ be the volume inside this light cone at time $t$.
Using the triangle inequality on \cref{eq:add-long-range}, we would arrive at
\begin{align}
	\norm{\P_r e^{\L t}\oket{O}} \lesssim \norm{\P_r e^{\L_n t}\oket{O}} + \frac{\mathcal V t}{\ell_n^{\alpha-d}},\label{eq:bound-with-V}
\end{align}
where $\mathcal V$ is the result of the sum over $i$ inside the light cone, summing over $j$ where $\dist(i,j)>\ell_n$ gives a factor proportional to $1/\ell_n^{\alpha-d}$, and
the integral over time  in \cref{eq:add-long-range} gives the factor $t$.

Suppose we can apply the desired light cone $t \gtrsim r^{\alpha-2d}$.
Then we can estimate the volume inside the light cone $\mathcal V \lesssim t^{d/(\alpha-2d)}$.
Substituting it into the above bound together with the value of $\ell_n$, we would arrive at
\begin{align}
	\norm{\P_r e^{\L t}\oket{O}} \lesssim \left(\frac{t}{r^{\alpha-2d}}\right)^{\frac{\alpha-d}{\alpha-2d }},\label{eq:limiting-case}
\end{align}
which gives about the same light cone as in \cref{thm:main-bound}.

However, we are proving \cref{thm:main-bound} and so cannot yet apply the light cone $t \gtrsim r^{\alpha-2d}$.
Instead, we use the light cone from Ref.~\cite{Foss-FeigG}, which is weaker than \cref{thm:main-bound}, to estimate $\mathcal V$.
Substituting this value of $\mathcal V$ into \cref{eq:bound-with-V}, we obtain a \emph{tighter} light cone than that of Ref.~\cite{Foss-FeigG}.
Iteratively using the resulting light cone to estimate $\mathcal V$ (see the SM \cite{SM} for a more detailed derivation), we obtain tighter and tighter bounds.
These bounds converge to a stable point that is exactly \cref{eq:limiting-case}.
Therefore, we obtain \cref{thm:main-bound}.

\section{Discussion}\Cref{thm:main-bound} implies a light cone that can be made arbitrarily close to $t \gtrsim r^{\alpha-2d}$ for all $\alpha \in (2d,2d+1)$.
In addition, \cref{thm:main-bound} also implies a linear light cone $t \gtrsim r^{1-o(1)}$ for $\alpha \geq 2d+1$, providing an alternative proof to Refs.~\cite{chenFiniteSpeedQuantum2019,kuwaharaStrictlyLinearLight2020} for two-body Hamiltonians.
Together with Refs.~\cite{HK,chenFiniteSpeedQuantum2019,kuwaharaStrictlyLinearLight2020}, we have the final Lieb-Robinson light cone for power-law interactions:
\begin{align}
	t \gtrsim \begin{cases}
		\log r & \text{ if } d < \alpha \leq 2d\\
		r^{\alpha-2d-o(1)} & \text{ if } 2d < \alpha \leq 2d+1\\
		r & \text{ if } \alpha > 2d+1
	\end{cases},
\end{align}
which we can saturate, up to subpolynomial corrections, using the protocol for state transfer and entanglement generation in Ref.~\cite{2020arXiv201002930T}.

Additionally, at any fixed time, our bound decays with distance as $1/r^{\alpha-d-o(1)}$.
Because the total strength of the interactions between the origin and all sites that are at distance at least $r$ from the origin already scales as $1/r^{\alpha-d}$, this so-called ``tail'' of our bound is also optimal.

Our result tightens the constraints on various quantum information tasks in power-law systems, including the growth of connected correlation functions, the generation of topological order, and the digital simulation of local observables.
Intuitively, as a local operator evolves, it is mostly constrained to lie within a light cone defined by a Lieb-Robinson bound, with total leakage outside this light cone constrained by the tail of this bound.
To simulate the dynamics of such observables, it is sufficient to simulate only the dynamics inside the light cone~\cite{tranLocalityDigitalQuantum2019a,PhysRevX.11.011020,tranHierarchyLinearLight2020a}, resulting in a more efficient simulation than simulating the entire lattice.
Similarly, the connected correlator between initially local observables remains small during the dynamics if their corresponding light cones have little overlap~\cite{BravyiHV06,GongFF,tranHierarchyLinearLight2020a}.
Topologically ordered states---those that cannot be distinguished by local observables---would also remain topologically ordered until local observables have enough time to substantially grow their supports~\cite{BravyiHV06,tranHierarchyLinearLight2020a}.
Crucially, then, \cref{thm:main-bound}, which has a provably optimal light cone and tail, provides the best-known asymptotic constraints for the dynamics of these quantities.
The mathematical details of precisely how they are bounded and the improvements that our new bound provides are detailed in the SM.

While we assume that the Hamiltonian is two-body throughout the paper, we expect the result extends to general many-body interactions. Specifically, we conjecture that \cref{thm:main-bound} holds for all Hamiltonians $H = \sum_{X\subset \Lambda} h_X$, where the sum is over all subsets of the lattice and $\sum_{X\ni i,j}\norm{h_X}\leq 1/\dist(i,j)^\alpha$ for all $i\neq j$.

Lastly, while Theorem 1 demonstrates the optimality of the single-particle state transfer protocol of \cite{2020arXiv201002930T},  other information-theoretic tasks are constrained by tighter light cones.  Our techniques may help extend recent progress \cite{tranHierarchyLinearLight2020a,kuwaharaPolynomialGrowthOutoftimeorder2020a,chen2021concentration} in constraining the remaining light cone hierarchy that has been demonstrated with power law interactions.

\begin{acknowledgments}
\section{Acknowledgments}We thank Abhinav Deshpande, Dhruv Devulapalli, Michael Foss-Feig, and Zhe-Xuan Gong for helpful discussions.
MCT, AYG, AE, and AVG acknowledge funding by the DoE ASCR Quantum Testbed Pathfinder program (award No.~DE-SC0019040), AFOSR MURI, NSF PFCQC program, AFOSR, DoE ASCR Accelerated Research in Quantum Computing program (award No.~DE-SC0020312), U.S. Department of Energy Award No.~DE-SC0019449, and ARO MURI.
MCT acknowledges additional support from the Princeton Center for Complex Materials, a MRSEC supported by NSF grant DMR 1420541.
AYG is supported by the NSF Graduate Research Fellowship Program under Grant No. DGE-1840340.
AL was supported by a Research Fellowship from the Alfred P. Sloan Foundation.
This research was performed while CLB~held an NRC Research Associateship award at the National Institute of Standards and Technology.

\end{acknowledgments}

\bibliography{my-bib}

\begin{thebibliography}{28}%
\makeatletter
\providecommand \@ifxundefined [1]{%
 \@ifx{#1\undefined}
}%
\providecommand \@ifnum [1]{%
 \ifnum #1\expandafter \@firstoftwo
 \else \expandafter \@secondoftwo
 \fi
}%
\providecommand \@ifx [1]{%
 \ifx #1\expandafter \@firstoftwo
 \else \expandafter \@secondoftwo
 \fi
}%
\providecommand \natexlab [1]{#1}%
\providecommand \enquote  [1]{``#1''}%
\providecommand \bibnamefont  [1]{#1}%
\providecommand \bibfnamefont [1]{#1}%
\providecommand \citenamefont [1]{#1}%
\providecommand \href@noop [0]{\@secondoftwo}%
\providecommand \href [0]{\begingroup \@sanitize@url \@href}%
\providecommand \@href[1]{\@@startlink{#1}\@@href}%
\providecommand \@@href[1]{\endgroup#1\@@endlink}%
\providecommand \@sanitize@url [0]{\catcode `\\12\catcode `\$12\catcode
  `\&12\catcode `\#12\catcode `\^12\catcode `\_12\catcode `\%12\relax}%
\providecommand \@@startlink[1]{}%
\providecommand \@@endlink[0]{}%
\providecommand \url  [0]{\begingroup\@sanitize@url \@url }%
\providecommand \@url [1]{\endgroup\@href {#1}{\urlprefix }}%
\providecommand \urlprefix  [0]{URL }%
\providecommand \Eprint [0]{\href }%
\providecommand \doibase [0]{https://doi.org/}%
\providecommand \selectlanguage [0]{\@gobble}%
\providecommand \bibinfo  [0]{\@secondoftwo}%
\providecommand \bibfield  [0]{\@secondoftwo}%
\providecommand \translation [1]{[#1]}%
\providecommand \BibitemOpen [0]{}%
\providecommand \bibitemStop [0]{}%
\providecommand \bibitemNoStop [0]{.\EOS\space}%
\providecommand \EOS [0]{\spacefactor3000\relax}%
\providecommand \BibitemShut  [1]{\csname bibitem#1\endcsname}%
\let\auto@bib@innerbib\@empty
\bibitem [{\citenamefont {Lieb}\ and\ \citenamefont {Robinson}(1972)}]{LR}%
  \BibitemOpen
  \bibfield  {author} {\bibinfo {author} {\bibfnamefont {E.~H.}\ \bibnamefont
  {Lieb}}\ and\ \bibinfo {author} {\bibfnamefont {D.~W.}\ \bibnamefont
  {Robinson}},\ }\bibfield  {title} {\bibinfo {title} {{The Finite Group
  Velocity of Quantum Spin Systems}},\ }\href
  {https://projecteuclid.org:443/euclid.cmp/1103858407} {\bibfield  {journal}
  {\bibinfo  {journal} {Comm. Math. Phys.}\ }\textbf {\bibinfo {volume} {28}},\
  \bibinfo {pages} {251} (\bibinfo {year} {1972})}\BibitemShut {NoStop}%
\bibitem [{\citenamefont {Wineland}\ \emph {et~al.}(1992)\citenamefont
  {Wineland}, \citenamefont {Bollinger}, \citenamefont {Itano}, \citenamefont
  {Moore},\ and\ \citenamefont {Heinzen}}]{PhysRevA.46.R6797}%
  \BibitemOpen
  \bibfield  {author} {\bibinfo {author} {\bibfnamefont {D.~J.}\ \bibnamefont
  {Wineland}}, \bibinfo {author} {\bibfnamefont {J.~J.}\ \bibnamefont
  {Bollinger}}, \bibinfo {author} {\bibfnamefont {W.~M.}\ \bibnamefont
  {Itano}}, \bibinfo {author} {\bibfnamefont {F.~L.}\ \bibnamefont {Moore}},\
  and\ \bibinfo {author} {\bibfnamefont {D.~J.}\ \bibnamefont {Heinzen}},\
  }\bibfield  {title} {\bibinfo {title} {Spin squeezing and reduced quantum
  noise in spectroscopy},\ }\href {https://doi.org/10.1103/PhysRevA.46.R6797}
  {\bibfield  {journal} {\bibinfo  {journal} {Phys. Rev. A}\ }\textbf {\bibinfo
  {volume} {46}},\ \bibinfo {pages} {R6797} (\bibinfo {year}
  {1992})}\BibitemShut {NoStop}%
\bibitem [{\citenamefont {Foss-Feig}\ \emph {et~al.}(2016)\citenamefont
  {Foss-Feig}, \citenamefont {Gong}, \citenamefont {Gorshkov},\ and\
  \citenamefont {Clark}}]{fossfeig2016entanglement}%
  \BibitemOpen
  \bibfield  {author} {\bibinfo {author} {\bibfnamefont {M.}~\bibnamefont
  {Foss-Feig}}, \bibinfo {author} {\bibfnamefont {Z.-X.}\ \bibnamefont {Gong}},
  \bibinfo {author} {\bibfnamefont {A.~V.}\ \bibnamefont {Gorshkov}},\ and\
  \bibinfo {author} {\bibfnamefont {C.~W.}\ \bibnamefont {Clark}},\ }\href@noop
  {} {\bibinfo {title} {Entanglement and spin-squeezing without infinite-range
  interactions}} (\bibinfo {year} {2016}),\ \Eprint
  {https://arxiv.org/abs/1612.07805} {arXiv:1612.07805} \BibitemShut {NoStop}%
\bibitem [{\citenamefont {Linke}\ \emph {et~al.}(2017)\citenamefont {Linke},
  \citenamefont {Maslov}, \citenamefont {Roetteler}, \citenamefont {Debnath},
  \citenamefont {Figgatt}, \citenamefont {Landsman}, \citenamefont {Wright},\
  and\ \citenamefont {Monroe}}]{linkeExperimentalComparisonTwo2017}%
  \BibitemOpen
  \bibfield  {author} {\bibinfo {author} {\bibfnamefont {N.~M.}\ \bibnamefont
  {Linke}}, \bibinfo {author} {\bibfnamefont {D.}~\bibnamefont {Maslov}},
  \bibinfo {author} {\bibfnamefont {M.}~\bibnamefont {Roetteler}}, \bibinfo
  {author} {\bibfnamefont {S.}~\bibnamefont {Debnath}}, \bibinfo {author}
  {\bibfnamefont {C.}~\bibnamefont {Figgatt}}, \bibinfo {author} {\bibfnamefont
  {K.~A.}\ \bibnamefont {Landsman}}, \bibinfo {author} {\bibfnamefont
  {K.}~\bibnamefont {Wright}},\ and\ \bibinfo {author} {\bibfnamefont
  {C.}~\bibnamefont {Monroe}},\ }\bibfield  {title} {\bibinfo {title}
  {Experimental comparison of two quantum computing architectures},\ }\href
  {https://doi.org/10.1073/pnas.1618020114} {\bibfield  {journal} {\bibinfo
  {journal} {Proc. Natl. Acad. Sci.}\ }\textbf {\bibinfo {volume} {114}},\
  \bibinfo {pages} {3305} (\bibinfo {year} {2017})}\BibitemShut {NoStop}%
\bibitem [{\citenamefont {Deshpande}\ \emph {et~al.}(2018)\citenamefont
  {Deshpande}, \citenamefont {Fefferman}, \citenamefont {Tran}, \citenamefont
  {{Foss-Feig}},\ and\ \citenamefont {Gorshkov}}]{Deshpande2018}%
  \BibitemOpen
  \bibfield  {author} {\bibinfo {author} {\bibfnamefont {A.}~\bibnamefont
  {Deshpande}}, \bibinfo {author} {\bibfnamefont {B.}~\bibnamefont
  {Fefferman}}, \bibinfo {author} {\bibfnamefont {M.~C.}\ \bibnamefont {Tran}},
  \bibinfo {author} {\bibfnamefont {M.}~\bibnamefont {{Foss-Feig}}},\ and\
  \bibinfo {author} {\bibfnamefont {A.~V.}\ \bibnamefont {Gorshkov}},\
  }\bibfield  {title} {\bibinfo {title} {Dynamical {{Phase Transitions}} in
  {{Sampling Complexity}}},\ }\href
  {https://doi.org/10.1103/PhysRevLett.121.030501} {\bibfield  {journal}
  {\bibinfo  {journal} {Phys. Rev. Lett.}\ }\textbf {\bibinfo {volume} {121}},\
  \bibinfo {pages} {030501} (\bibinfo {year} {2018})}\BibitemShut {NoStop}%
\bibitem [{\citenamefont {Landsman}\ \emph {et~al.}(2019)\citenamefont
  {Landsman}, \citenamefont {Figgatt}, \citenamefont {Schuster}, \citenamefont
  {Linke}, \citenamefont {Yoshida}, \citenamefont {Yao},\ and\ \citenamefont
  {Monroe}}]{landsmanVerifiedQuantumInformation2019}%
  \BibitemOpen
  \bibfield  {author} {\bibinfo {author} {\bibfnamefont {K.~A.}\ \bibnamefont
  {Landsman}}, \bibinfo {author} {\bibfnamefont {C.}~\bibnamefont {Figgatt}},
  \bibinfo {author} {\bibfnamefont {T.}~\bibnamefont {Schuster}}, \bibinfo
  {author} {\bibfnamefont {N.~M.}\ \bibnamefont {Linke}}, \bibinfo {author}
  {\bibfnamefont {B.}~\bibnamefont {Yoshida}}, \bibinfo {author} {\bibfnamefont
  {N.~Y.}\ \bibnamefont {Yao}},\ and\ \bibinfo {author} {\bibfnamefont
  {C.}~\bibnamefont {Monroe}},\ }\bibfield  {title} {\bibinfo {title} {Verified
  quantum information scrambling},\ }\href
  {https://doi.org/10.1038/s41586-019-0952-6} {\bibfield  {journal} {\bibinfo
  {journal} {Nature}\ }\textbf {\bibinfo {volume} {567}},\ \bibinfo {pages}
  {61} (\bibinfo {year} {2019})}\BibitemShut {NoStop}%
\bibitem [{\citenamefont {{Hastings}}\ and\ \citenamefont {{Koma}}(2006)}]{HK}%
  \BibitemOpen
  \bibfield  {author} {\bibinfo {author} {\bibfnamefont {M.~B.}\ \bibnamefont
  {{Hastings}}}\ and\ \bibinfo {author} {\bibfnamefont {T.}~\bibnamefont
  {{Koma}}},\ }\bibfield  {title} {\bibinfo {title} {{Spectral Gap and
  Exponential Decay of Correlations}},\ }\href
  {https://doi.org/10.1007/s00220-006-0030-4} {\bibfield  {journal} {\bibinfo
  {journal} {Comm. Math. Phys.}\ }\textbf {\bibinfo {volume} {265}},\ \bibinfo
  {pages} {781} (\bibinfo {year} {2006})}\BibitemShut {NoStop}%
\bibitem [{\citenamefont {Foss-Feig}\ \emph {et~al.}(2015)\citenamefont
  {Foss-Feig}, \citenamefont {Gong}, \citenamefont {Clark},\ and\ \citenamefont
  {Gorshkov}}]{Foss-FeigG}%
  \BibitemOpen
  \bibfield  {author} {\bibinfo {author} {\bibfnamefont {M.}~\bibnamefont
  {Foss-Feig}}, \bibinfo {author} {\bibfnamefont {Z.-X.}\ \bibnamefont {Gong}},
  \bibinfo {author} {\bibfnamefont {C.~W.}\ \bibnamefont {Clark}},\ and\
  \bibinfo {author} {\bibfnamefont {A.~V.}\ \bibnamefont {Gorshkov}},\
  }\bibfield  {title} {\bibinfo {title} {{Nearly Linear Light Cones in
  Long-Range Interacting Quantum Systems}},\ }\href
  {https://doi.org/10.1103/PhysRevLett.114.157201} {\bibfield  {journal}
  {\bibinfo  {journal} {Phys. Rev. Lett.}\ }\textbf {\bibinfo {volume} {114}},\
  \bibinfo {pages} {157201} (\bibinfo {year} {2015})}\BibitemShut {NoStop}%
\bibitem [{\citenamefont {Else}\ \emph {et~al.}(2020)\citenamefont {Else},
  \citenamefont {Machado}, \citenamefont {Nayak},\ and\ \citenamefont
  {Yao}}]{elseImprovedLiebRobinsonBound2018}%
  \BibitemOpen
  \bibfield  {author} {\bibinfo {author} {\bibfnamefont {D.~V.}\ \bibnamefont
  {Else}}, \bibinfo {author} {\bibfnamefont {F.}~\bibnamefont {Machado}},
  \bibinfo {author} {\bibfnamefont {C.}~\bibnamefont {Nayak}},\ and\ \bibinfo
  {author} {\bibfnamefont {N.~Y.}\ \bibnamefont {Yao}},\ }\bibfield  {title}
  {\bibinfo {title} {Improved lieb-robinson bound for many-body hamiltonians
  with power-law interactions},\ }\href
  {https://doi.org/10.1103/PhysRevA.101.022333} {\bibfield  {journal} {\bibinfo
   {journal} {Phys. Rev. A}\ }\textbf {\bibinfo {volume} {101}},\ \bibinfo
  {pages} {022333} (\bibinfo {year} {2020})}\BibitemShut {NoStop}%
\bibitem [{\citenamefont {Tran}\ \emph {et~al.}(2019)\citenamefont {Tran},
  \citenamefont {Guo}, \citenamefont {Su}, \citenamefont {Garrison},
  \citenamefont {Eldredge}, \citenamefont {{Foss-Feig}}, \citenamefont
  {Childs},\ and\ \citenamefont {Gorshkov}}]{tranLocalityDigitalQuantum2019a}%
  \BibitemOpen
  \bibfield  {author} {\bibinfo {author} {\bibfnamefont {M.~C.}\ \bibnamefont
  {Tran}}, \bibinfo {author} {\bibfnamefont {A.~Y.}\ \bibnamefont {Guo}},
  \bibinfo {author} {\bibfnamefont {Y.}~\bibnamefont {Su}}, \bibinfo {author}
  {\bibfnamefont {J.~R.}\ \bibnamefont {Garrison}}, \bibinfo {author}
  {\bibfnamefont {Z.}~\bibnamefont {Eldredge}}, \bibinfo {author}
  {\bibfnamefont {M.}~\bibnamefont {{Foss-Feig}}}, \bibinfo {author}
  {\bibfnamefont {A.~M.}\ \bibnamefont {Childs}},\ and\ \bibinfo {author}
  {\bibfnamefont {A.~V.}\ \bibnamefont {Gorshkov}},\ }\bibfield  {title}
  {\bibinfo {title} {Locality and {{Digital Quantum Simulation}} of
  {{Power}}-{{Law Interactions}}},\ }\href
  {https://doi.org/10.1103/PhysRevX.9.031006} {\bibfield  {journal} {\bibinfo
  {journal} {Phys. Rev. X}\ }\textbf {\bibinfo {volume} {9}},\ \bibinfo {pages}
  {031006} (\bibinfo {year} {2019})}\BibitemShut {NoStop}%
\bibitem [{\citenamefont {Chen}\ and\ \citenamefont
  {Lucas}(2019{\natexlab{a}})}]{chenFiniteSpeedQuantum2019}%
  \BibitemOpen
  \bibfield  {author} {\bibinfo {author} {\bibfnamefont {C.-F.}\ \bibnamefont
  {Chen}}\ and\ \bibinfo {author} {\bibfnamefont {A.}~\bibnamefont {Lucas}},\
  }\bibfield  {title} {\bibinfo {title} {Finite speed of quantum scrambling
  with long range interactions},\ }\href
  {https://doi.org/10.1103/PhysRevLett.123.250605} {\bibfield  {journal}
  {\bibinfo  {journal} {Phys. Rev. Lett.}\ }\textbf {\bibinfo {volume} {123}},\
  \bibinfo {pages} {250605} (\bibinfo {year} {2019}{\natexlab{a}})}\BibitemShut
  {NoStop}%
\bibitem [{\citenamefont {Kuwahara}\ and\ \citenamefont
  {Saito}(2020{\natexlab{a}})}]{kuwaharaStrictlyLinearLight2020}%
  \BibitemOpen
  \bibfield  {author} {\bibinfo {author} {\bibfnamefont {T.}~\bibnamefont
  {Kuwahara}}\ and\ \bibinfo {author} {\bibfnamefont {K.}~\bibnamefont
  {Saito}},\ }\bibfield  {title} {\bibinfo {title} {Strictly linear light cones
  in long-range interacting systems of arbitrary dimensions},\ }\href
  {https://doi.org/10.1103/PhysRevX.10.031010} {\bibfield  {journal} {\bibinfo
  {journal} {Phys. Rev. X}\ }\textbf {\bibinfo {volume} {10}},\ \bibinfo
  {pages} {031010} (\bibinfo {year} {2020}{\natexlab{a}})}\BibitemShut
  {NoStop}%
\bibitem [{\citenamefont {Eldredge}\ \emph {et~al.}()\citenamefont {Eldredge},
  \citenamefont {Gong}, \citenamefont {Young}, \citenamefont {Moosavian},
  \citenamefont {Foss-Feig},\ and\ \citenamefont {Gorshkov}}]{Zachary17}%
  \BibitemOpen
  \bibfield  {author} {\bibinfo {author} {\bibfnamefont {Z.}~\bibnamefont
  {Eldredge}}, \bibinfo {author} {\bibfnamefont {Z.-X.}\ \bibnamefont {Gong}},
  \bibinfo {author} {\bibfnamefont {J.~T.}\ \bibnamefont {Young}}, \bibinfo
  {author} {\bibfnamefont {A.~H.}\ \bibnamefont {Moosavian}}, \bibinfo {author}
  {\bibfnamefont {M.}~\bibnamefont {Foss-Feig}},\ and\ \bibinfo {author}
  {\bibfnamefont {A.~V.}\ \bibnamefont {Gorshkov}},\ }\bibfield  {title}
  {\bibinfo {title} {Fast quantum state transfer and entanglement
  renormalization using long-range interactions},\ }\href
  {https://doi.org/10.1103/PhysRevLett.119.170503} {\bibinfo  {journal} {Phys.
  Rev. Lett.}\ ,\ \bibinfo {pages} {170503}}\BibitemShut {NoStop}%
\bibitem [{\citenamefont {Tran}\ \emph {et~al.}(2020)\citenamefont {Tran},
  \citenamefont {Chen}, \citenamefont {Ehrenberg}, \citenamefont {Guo},
  \citenamefont {Deshpande}, \citenamefont {Hong}, \citenamefont {Gong},
  \citenamefont {Gorshkov},\ and\ \citenamefont
  {Lucas}}]{tranHierarchyLinearLight2020a}%
  \BibitemOpen
\bibfield  {journal} {  }\bibfield  {author} {\bibinfo {author} {\bibfnamefont
  {M.~C.}\ \bibnamefont {Tran}}, \bibinfo {author} {\bibfnamefont {C.-F.}\
  \bibnamefont {Chen}}, \bibinfo {author} {\bibfnamefont {A.}~\bibnamefont
  {Ehrenberg}}, \bibinfo {author} {\bibfnamefont {A.~Y.}\ \bibnamefont {Guo}},
  \bibinfo {author} {\bibfnamefont {A.}~\bibnamefont {Deshpande}}, \bibinfo
  {author} {\bibfnamefont {Y.}~\bibnamefont {Hong}}, \bibinfo {author}
  {\bibfnamefont {Z.-X.}\ \bibnamefont {Gong}}, \bibinfo {author}
  {\bibfnamefont {A.~V.}\ \bibnamefont {Gorshkov}},\ and\ \bibinfo {author}
  {\bibfnamefont {A.}~\bibnamefont {Lucas}},\ }\bibfield  {title} {\bibinfo
  {title} {Hierarchy of {{Linear Light Cones}} with {{Long}}-{{Range
  Interactions}}},\ }\href {https://doi.org/10.1103/PhysRevX.10.031009}
  {\bibfield  {journal} {\bibinfo  {journal} {Phys. Rev. X}\ }\textbf {\bibinfo
  {volume} {10}},\ \bibinfo {pages} {031009} (\bibinfo {year}
  {2020})}\BibitemShut {NoStop}%
\bibitem [{\citenamefont {{Tran}}\ \emph {et~al.}(2020)\citenamefont {{Tran}},
  \citenamefont {{Deshpande}}, \citenamefont {{Guo}}, \citenamefont {{Lucas}},\
  and\ \citenamefont {{Gorshkov}}}]{2020arXiv201002930T}%
  \BibitemOpen
  \bibfield  {author} {\bibinfo {author} {\bibfnamefont {M.~C.}\ \bibnamefont
  {{Tran}}}, \bibinfo {author} {\bibfnamefont {A.}~\bibnamefont {{Deshpande}}},
  \bibinfo {author} {\bibfnamefont {A.~Y.}\ \bibnamefont {{Guo}}}, \bibinfo
  {author} {\bibfnamefont {A.}~\bibnamefont {{Lucas}}},\ and\ \bibinfo {author}
  {\bibfnamefont {A.~V.}\ \bibnamefont {{Gorshkov}}},\ }\href@noop {} {\bibinfo
  {title} {{Optimal State Transfer and Entanglement Generation in Power-law
  Interacting Systems}}} (\bibinfo {year} {2020}),\ \Eprint
  {https://arxiv.org/abs/2010.02930} {arXiv:2010.02930 [quant-ph]} \BibitemShut
  {NoStop}%
\bibitem [{\citenamefont {Nachtergaele}\ \emph {et~al.}(2006)\citenamefont
  {Nachtergaele}, \citenamefont {Ogata},\ and\ \citenamefont
  {Sims}}]{NachtergaeleOS2006}%
  \BibitemOpen
  \bibfield  {author} {\bibinfo {author} {\bibfnamefont {B.}~\bibnamefont
  {Nachtergaele}}, \bibinfo {author} {\bibfnamefont {Y.}~\bibnamefont
  {Ogata}},\ and\ \bibinfo {author} {\bibfnamefont {R.}~\bibnamefont {Sims}},\
  }\bibfield  {title} {\bibinfo {title} {{Propagation of Correlations in
  Quantum Lattice Systems}},\ }\href
  {https://doi.org/10.1007/s10955-006-9143-6} {\bibfield  {journal} {\bibinfo
  {journal} {J. Stat. Phys.}\ }\textbf {\bibinfo {volume} {124}},\ \bibinfo
  {pages} {1} (\bibinfo {year} {2006})}\BibitemShut {NoStop}%
\bibitem [{\citenamefont {Nachtergaele}\ and\ \citenamefont
  {Sims}(2006)}]{Nachtergaele2006}%
  \BibitemOpen
  \bibfield  {author} {\bibinfo {author} {\bibfnamefont {B.}~\bibnamefont
  {Nachtergaele}}\ and\ \bibinfo {author} {\bibfnamefont {R.}~\bibnamefont
  {Sims}},\ }\bibfield  {title} {\bibinfo {title} {{Lieb-Robinson Bounds and
  the Exponential Clustering Theorem}},\ }\href
  {https://doi.org/10.1007/s00220-006-1556-1} {\bibfield  {journal} {\bibinfo
  {journal} {Comm. Math. Phys.}\ }\textbf {\bibinfo {volume} {265}},\ \bibinfo
  {pages} {119} (\bibinfo {year} {2006})}\BibitemShut {NoStop}%
\bibitem [{\citenamefont {{Gong}}\ \emph {et~al.}(2014)\citenamefont {{Gong}},
  \citenamefont {{Foss-Feig}}, \citenamefont {{Michalakis}},\ and\
  \citenamefont {{Gorshkov}}}]{GongFF}%
  \BibitemOpen
  \bibfield  {author} {\bibinfo {author} {\bibfnamefont {Z.-X.}\ \bibnamefont
  {{Gong}}}, \bibinfo {author} {\bibfnamefont {M.}~\bibnamefont {{Foss-Feig}}},
  \bibinfo {author} {\bibfnamefont {S.}~\bibnamefont {{Michalakis}}},\ and\
  \bibinfo {author} {\bibfnamefont {A.~V.}\ \bibnamefont {{Gorshkov}}},\
  }\bibfield  {title} {\bibinfo {title} {{Persistence of Locality in Systems
  With Power-Law Interactions}},\ }\href
  {https://doi.org/10.1103/PhysRevLett.113.030602} {\bibfield  {journal}
  {\bibinfo  {journal} {Phys. Rev. Lett.}\ }\textbf {\bibinfo {volume} {113}},\
  \bibinfo {eid} {030602} (\bibinfo {year} {2014})}\BibitemShut {NoStop}%
\bibitem [{\citenamefont {Storch}\ \emph {et~al.}(2015)\citenamefont {Storch},
  \citenamefont {Worm},\ and\ \citenamefont {Kastner}}]{Storch15}%
  \BibitemOpen
  \bibfield  {author} {\bibinfo {author} {\bibfnamefont {D.-M.}\ \bibnamefont
  {Storch}}, \bibinfo {author} {\bibfnamefont {M.~V.~D.}\ \bibnamefont
  {Worm}},\ and\ \bibinfo {author} {\bibfnamefont {M.}~\bibnamefont
  {Kastner}},\ }\bibfield  {title} {\bibinfo {title} {{Interplay of Soundcone
  and Supersonic Propagation in Lattice Models With Power Law Interactions}},\
  }\href {http://stacks.iop.org/1367-2630/17/i=6/a=063021} {\bibfield
  {journal} {\bibinfo  {journal} {New J. Phys.}\ }\textbf {\bibinfo {volume}
  {17}},\ \bibinfo {pages} {063021} (\bibinfo {year} {2015})}\BibitemShut
  {NoStop}%
\bibitem [{\citenamefont {Nachtergaele}\ \emph {et~al.}(2009)\citenamefont
  {Nachtergaele}, \citenamefont {Raz}, \citenamefont {Schlein},\ and\
  \citenamefont {Sims}}]{NRSS09}%
  \BibitemOpen
  \bibfield  {author} {\bibinfo {author} {\bibfnamefont {B.}~\bibnamefont
  {Nachtergaele}}, \bibinfo {author} {\bibfnamefont {H.}~\bibnamefont {Raz}},
  \bibinfo {author} {\bibfnamefont {B.}~\bibnamefont {Schlein}},\ and\ \bibinfo
  {author} {\bibfnamefont {R.}~\bibnamefont {Sims}},\ }\bibfield  {title}
  {\bibinfo {title} {Lieb-robinson bounds for harmonic and anharmonic lattice
  systems},\ }\href {https://doi.org/10.1007/s00220-008-0630-2} {\bibfield
  {journal} {\bibinfo  {journal} {Comm. Math. Phys.}\ }\textbf {\bibinfo
  {volume} {286}},\ \bibinfo {pages} {1073} (\bibinfo {year}
  {2009})}\BibitemShut {NoStop}%
\bibitem [{\citenamefont {Pr\'emont-Schwarz}\ \emph {et~al.}(2010)\citenamefont
  {Pr\'emont-Schwarz}, \citenamefont {Hamma}, \citenamefont {Klich},\ and\
  \citenamefont {Markopoulou-Kalamara}}]{SHKM10}%
  \BibitemOpen
  \bibfield  {author} {\bibinfo {author} {\bibfnamefont {I.}~\bibnamefont
  {Pr\'emont-Schwarz}}, \bibinfo {author} {\bibfnamefont {A.}~\bibnamefont
  {Hamma}}, \bibinfo {author} {\bibfnamefont {I.}~\bibnamefont {Klich}},\ and\
  \bibinfo {author} {\bibfnamefont {F.}~\bibnamefont {Markopoulou-Kalamara}},\
  }\bibfield  {title} {\bibinfo {title} {Lieb-robinson bounds for
  commutator-bounded operators},\ }\href
  {https://doi.org/10.1103/PhysRevA.81.040102} {\bibfield  {journal} {\bibinfo
  {journal} {Phys. Rev. A}\ }\textbf {\bibinfo {volume} {81}},\ \bibinfo
  {pages} {040102} (\bibinfo {year} {2010})}\BibitemShut {NoStop}%
\bibitem [{\citenamefont {Pr\'emont-Schwarz}\ and\ \citenamefont
  {Hnybida}(2010)}]{SH10}%
  \BibitemOpen
  \bibfield  {author} {\bibinfo {author} {\bibfnamefont {I.}~\bibnamefont
  {Pr\'emont-Schwarz}}\ and\ \bibinfo {author} {\bibfnamefont {J.}~\bibnamefont
  {Hnybida}},\ }\bibfield  {title} {\bibinfo {title} {Lieb-robinson bounds on
  the speed of information propagation},\ }\href
  {https://doi.org/10.1103/PhysRevA.81.062107} {\bibfield  {journal} {\bibinfo
  {journal} {Phys. Rev. A}\ }\textbf {\bibinfo {volume} {81}},\ \bibinfo
  {pages} {062107} (\bibinfo {year} {2010})}\BibitemShut {NoStop}%
\bibitem [{SM()}]{SM}%
  \BibitemOpen
  \href@noop {} {}\bibinfo {note} {In the Supplemental Material, we provide a
  rigorous proof of \cref{thm:main-bound} and mathematical details for the
  applications of the bound.}\BibitemShut {Stop}%
\bibitem [{\citenamefont {Chen}\ and\ \citenamefont
  {Lucas}(2019{\natexlab{b}})}]{chenOperatorGrowthBounds2019}%
  \BibitemOpen
  \bibfield  {author} {\bibinfo {author} {\bibfnamefont {C.-F.}\ \bibnamefont
  {Chen}}\ and\ \bibinfo {author} {\bibfnamefont {A.}~\bibnamefont {Lucas}},\
  }\bibfield  {title} {\bibinfo {title} {Operator growth bounds from graph
  theory},\ }\href@noop {} {\bibfield  {journal} {\bibinfo  {journal}
  {arXiv:1905.03682 [hep-th, physics:math-ph, physics:quant-ph]}\ } (\bibinfo
  {year} {2019}{\natexlab{b}})},\ \Eprint {https://arxiv.org/abs/1905.03682}
  {arXiv:1905.03682 [hep-th, physics:math-ph, physics:quant-ph]} \BibitemShut
  {NoStop}%
\bibitem [{\citenamefont {Childs}\ \emph {et~al.}(2021)\citenamefont {Childs},
  \citenamefont {Su}, \citenamefont {Tran}, \citenamefont {Wiebe},\ and\
  \citenamefont {Zhu}}]{PhysRevX.11.011020}%
  \BibitemOpen
  \bibfield  {author} {\bibinfo {author} {\bibfnamefont {A.~M.}\ \bibnamefont
  {Childs}}, \bibinfo {author} {\bibfnamefont {Y.}~\bibnamefont {Su}}, \bibinfo
  {author} {\bibfnamefont {M.~C.}\ \bibnamefont {Tran}}, \bibinfo {author}
  {\bibfnamefont {N.}~\bibnamefont {Wiebe}},\ and\ \bibinfo {author}
  {\bibfnamefont {S.}~\bibnamefont {Zhu}},\ }\bibfield  {title} {\bibinfo
  {title} {Theory of trotter error with commutator scaling},\ }\href
  {https://doi.org/10.1103/PhysRevX.11.011020} {\bibfield  {journal} {\bibinfo
  {journal} {Phys. Rev. X}\ }\textbf {\bibinfo {volume} {11}},\ \bibinfo
  {pages} {011020} (\bibinfo {year} {2021})}\BibitemShut {NoStop}%
\bibitem [{\citenamefont {Bravyi}\ \emph {et~al.}(2006)\citenamefont {Bravyi},
  \citenamefont {Hastings},\ and\ \citenamefont {Verstraete}}]{BravyiHV06}%
  \BibitemOpen
  \bibfield  {author} {\bibinfo {author} {\bibfnamefont {S.}~\bibnamefont
  {Bravyi}}, \bibinfo {author} {\bibfnamefont {M.~B.}\ \bibnamefont
  {Hastings}},\ and\ \bibinfo {author} {\bibfnamefont {F.}~\bibnamefont
  {Verstraete}},\ }\bibfield  {title} {\bibinfo {title} {{Lieb-Robinson Bounds
  and the Generation of Correlations and Topological Quantum Order}},\ }\href
  {https://doi.org/10.1103/PhysRevLett.97.050401} {\bibfield  {journal}
  {\bibinfo  {journal} {Phys. Rev. Lett.}\ }\textbf {\bibinfo {volume} {97}},\
  \bibinfo {pages} {050401} (\bibinfo {year} {2006})}\BibitemShut {NoStop}%
\bibitem [{\citenamefont {Kuwahara}\ and\ \citenamefont
  {Saito}(2020{\natexlab{b}})}]{kuwaharaPolynomialGrowthOutoftimeorder2020a}%
  \BibitemOpen
  \bibfield  {author} {\bibinfo {author} {\bibfnamefont {T.}~\bibnamefont
  {Kuwahara}}\ and\ \bibinfo {author} {\bibfnamefont {K.}~\bibnamefont
  {Saito}},\ }\bibfield  {title} {\bibinfo {title} {Polynomial growth of
  out-of-time-order correlator in arbitrary realistic long-range interacting
  systems},\ }\href {http://arxiv.org/abs/2009.10124} {\bibfield  {journal}
  {\bibinfo  {journal} {arXiv:2009.10124}\ } (\bibinfo {year}
  {2020}{\natexlab{b}})}\BibitemShut {NoStop}%
\bibitem [{\citenamefont {Chen}(2021)}]{chen2021concentration}%
  \BibitemOpen
  \bibfield  {author} {\bibinfo {author} {\bibfnamefont {C.-F.}\ \bibnamefont
  {Chen}},\ }\href@noop {} {\bibinfo {title} {Concentration of otoc and
  lieb-robinson velocity in random hamiltonians}} (\bibinfo {year} {2021}),\
  \Eprint {https://arxiv.org/abs/2103.09186} {arXiv:2103.09186 [quant-ph]}
  \BibitemShut {NoStop}%
\end{thebibliography}%


\begin{thebibliography}{7}%
\makeatletter
\providecommand \@ifxundefined [1]{%
 \@ifx{#1\undefined}
}%
\providecommand \@ifnum [1]{%
 \ifnum #1\expandafter \@firstoftwo
 \else \expandafter \@secondoftwo
 \fi
}%
\providecommand \@ifx [1]{%
 \ifx #1\expandafter \@firstoftwo
 \else \expandafter \@secondoftwo
 \fi
}%
\providecommand \natexlab [1]{#1}%
\providecommand \enquote  [1]{``#1''}%
\providecommand \bibnamefont  [1]{#1}%
\providecommand \bibfnamefont [1]{#1}%
\providecommand \citenamefont [1]{#1}%
\providecommand \href@noop [0]{\@secondoftwo}%
\providecommand \href [0]{\begingroup \@sanitize@url \@href}%
\providecommand \@href[1]{\@@startlink{#1}\@@href}%
\providecommand \@@href[1]{\endgroup#1\@@endlink}%
\providecommand \@sanitize@url [0]{\catcode `\\12\catcode `\$12\catcode
  `\&12\catcode `\#12\catcode `\^12\catcode `\_12\catcode `\%12\relax}%
\providecommand \@@startlink[1]{}%
\providecommand \@@endlink[0]{}%
\providecommand \url  [0]{\begingroup\@sanitize@url \@url }%
\providecommand \@url [1]{\endgroup\@href {#1}{\urlprefix }}%
\providecommand \urlprefix  [0]{URL }%
\providecommand \Eprint [0]{\href }%
\providecommand \doibase [0]{https://doi.org/}%
\providecommand \selectlanguage [0]{\@gobble}%
\providecommand \bibinfo  [0]{\@secondoftwo}%
\providecommand \bibfield  [0]{\@secondoftwo}%
\providecommand \translation [1]{[#1]}%
\providecommand \BibitemOpen [0]{}%
\providecommand \bibitemStop [0]{}%
\providecommand \bibitemNoStop [0]{.\EOS\space}%
\providecommand \EOS [0]{\spacefactor3000\relax}%
\providecommand \BibitemShut  [1]{\csname bibitem#1\endcsname}%
\let\auto@bib@innerbib\@empty
\bibitem [{\citenamefont {Chen}\ and\ \citenamefont
  {Lucas}(2019{\natexlab{a}})}]{chenFiniteSpeedQuantum2019}%
  \BibitemOpen
  \bibfield  {author} {\bibinfo {author} {\bibfnamefont {C.-F.}\ \bibnamefont
  {Chen}}\ and\ \bibinfo {author} {\bibfnamefont {A.}~\bibnamefont {Lucas}},\
  }\bibfield  {title} {\bibinfo {title} {Finite speed of quantum scrambling
  with long range interactions},\ }\href
  {https://doi.org/10.1103/PhysRevLett.123.250605} {\bibfield  {journal}
  {\bibinfo  {journal} {Phys. Rev. Lett.}\ }\textbf {\bibinfo {volume} {123}},\
  \bibinfo {pages} {250605} (\bibinfo {year} {2019}{\natexlab{a}})}\BibitemShut
  {NoStop}%
\bibitem [{\citenamefont {{Hastings}}\ and\ \citenamefont {{Koma}}(2006)}]{HK}%
  \BibitemOpen
  \bibfield  {author} {\bibinfo {author} {\bibfnamefont {M.~B.}\ \bibnamefont
  {{Hastings}}}\ and\ \bibinfo {author} {\bibfnamefont {T.}~\bibnamefont
  {{Koma}}},\ }\bibfield  {title} {\bibinfo {title} {{Spectral Gap and
  Exponential Decay of Correlations}},\ }\href
  {https://doi.org/10.1007/s00220-006-0030-4} {\bibfield  {journal} {\bibinfo
  {journal} {Comm. Math. Phys.}\ }\textbf {\bibinfo {volume} {265}},\ \bibinfo
  {pages} {781} (\bibinfo {year} {2006})}\BibitemShut {NoStop}%
\bibitem [{\citenamefont {Chen}\ and\ \citenamefont
  {Lucas}(2019{\natexlab{b}})}]{chenOperatorGrowthBounds2019}%
  \BibitemOpen
  \bibfield  {author} {\bibinfo {author} {\bibfnamefont {C.-F.}\ \bibnamefont
  {Chen}}\ and\ \bibinfo {author} {\bibfnamefont {A.}~\bibnamefont {Lucas}},\
  }\bibfield  {title} {\bibinfo {title} {Operator growth bounds from graph
  theory},\ }\href@noop {} {\bibfield  {journal} {\bibinfo  {journal}
  {arXiv:1905.03682 [hep-th, physics:math-ph, physics:quant-ph]}\ } (\bibinfo
  {year} {2019}{\natexlab{b}})},\ \Eprint {https://arxiv.org/abs/1905.03682}
  {arXiv:1905.03682 [hep-th, physics:math-ph, physics:quant-ph]} \BibitemShut
  {NoStop}%
\bibitem [{\citenamefont {Tran}\ \emph {et~al.}(2020)\citenamefont {Tran},
  \citenamefont {Chen}, \citenamefont {Ehrenberg}, \citenamefont {Guo},
  \citenamefont {Deshpande}, \citenamefont {Hong}, \citenamefont {Gong},
  \citenamefont {Gorshkov},\ and\ \citenamefont
  {Lucas}}]{tranHierarchyLinearLight2020a}%
  \BibitemOpen
  \bibfield  {author} {\bibinfo {author} {\bibfnamefont {M.~C.}\ \bibnamefont
  {Tran}}, \bibinfo {author} {\bibfnamefont {C.-F.}\ \bibnamefont {Chen}},
  \bibinfo {author} {\bibfnamefont {A.}~\bibnamefont {Ehrenberg}}, \bibinfo
  {author} {\bibfnamefont {A.~Y.}\ \bibnamefont {Guo}}, \bibinfo {author}
  {\bibfnamefont {A.}~\bibnamefont {Deshpande}}, \bibinfo {author}
  {\bibfnamefont {Y.}~\bibnamefont {Hong}}, \bibinfo {author} {\bibfnamefont
  {Z.-X.}\ \bibnamefont {Gong}}, \bibinfo {author} {\bibfnamefont {A.~V.}\
  \bibnamefont {Gorshkov}},\ and\ \bibinfo {author} {\bibfnamefont
  {A.}~\bibnamefont {Lucas}},\ }\bibfield  {title} {\bibinfo {title} {Hierarchy
  of {{Linear Light Cones}} with {{Long}}-{{Range Interactions}}},\ }\href
  {https://doi.org/10.1103/PhysRevX.10.031009} {\bibfield  {journal} {\bibinfo
  {journal} {Phys. Rev. X}\ }\textbf {\bibinfo {volume} {10}},\ \bibinfo
  {pages} {031009} (\bibinfo {year} {2020})}\BibitemShut {NoStop}%
\bibitem [{\citenamefont {Foss-Feig}\ \emph {et~al.}(2015)\citenamefont
  {Foss-Feig}, \citenamefont {Gong}, \citenamefont {Clark},\ and\ \citenamefont
  {Gorshkov}}]{Foss-FeigG}%
  \BibitemOpen
  \bibfield  {author} {\bibinfo {author} {\bibfnamefont {M.}~\bibnamefont
  {Foss-Feig}}, \bibinfo {author} {\bibfnamefont {Z.-X.}\ \bibnamefont {Gong}},
  \bibinfo {author} {\bibfnamefont {C.~W.}\ \bibnamefont {Clark}},\ and\
  \bibinfo {author} {\bibfnamefont {A.~V.}\ \bibnamefont {Gorshkov}},\
  }\bibfield  {title} {\bibinfo {title} {{Nearly Linear Light Cones in
  Long-Range Interacting Quantum Systems}},\ }\href
  {https://doi.org/10.1103/PhysRevLett.114.157201} {\bibfield  {journal}
  {\bibinfo  {journal} {Phys. Rev. Lett.}\ }\textbf {\bibinfo {volume} {114}},\
  \bibinfo {pages} {157201} (\bibinfo {year} {2015})}\BibitemShut {NoStop}%
\bibitem [{\citenamefont {Tran}\ \emph {et~al.}(2019)\citenamefont {Tran},
  \citenamefont {Guo}, \citenamefont {Su}, \citenamefont {Garrison},
  \citenamefont {Eldredge}, \citenamefont {{Foss-Feig}}, \citenamefont
  {Childs},\ and\ \citenamefont {Gorshkov}}]{tranLocalityDigitalQuantum2019a}%
  \BibitemOpen
  \bibfield  {author} {\bibinfo {author} {\bibfnamefont {M.~C.}\ \bibnamefont
  {Tran}}, \bibinfo {author} {\bibfnamefont {A.~Y.}\ \bibnamefont {Guo}},
  \bibinfo {author} {\bibfnamefont {Y.}~\bibnamefont {Su}}, \bibinfo {author}
  {\bibfnamefont {J.~R.}\ \bibnamefont {Garrison}}, \bibinfo {author}
  {\bibfnamefont {Z.}~\bibnamefont {Eldredge}}, \bibinfo {author}
  {\bibfnamefont {M.}~\bibnamefont {{Foss-Feig}}}, \bibinfo {author}
  {\bibfnamefont {A.~M.}\ \bibnamefont {Childs}},\ and\ \bibinfo {author}
  {\bibfnamefont {A.~V.}\ \bibnamefont {Gorshkov}},\ }\bibfield  {title}
  {\bibinfo {title} {Locality and {{Digital Quantum Simulation}} of
  {{Power}}-{{Law Interactions}}},\ }\href
  {https://doi.org/10.1103/PhysRevX.9.031006} {\bibfield  {journal} {\bibinfo
  {journal} {Phys. Rev. X}\ }\textbf {\bibinfo {volume} {9}},\ \bibinfo {pages}
  {031006} (\bibinfo {year} {2019})}\BibitemShut {NoStop}%
\bibitem [{\citenamefont {Bravyi}\ \emph {et~al.}(2006)\citenamefont {Bravyi},
  \citenamefont {Hastings},\ and\ \citenamefont {Verstraete}}]{Bravyi06}%
  \BibitemOpen
  \bibfield  {author} {\bibinfo {author} {\bibfnamefont {S.}~\bibnamefont
  {Bravyi}}, \bibinfo {author} {\bibfnamefont {M.~B.}\ \bibnamefont
  {Hastings}},\ and\ \bibinfo {author} {\bibfnamefont {F.}~\bibnamefont
  {Verstraete}},\ }\bibfield  {title} {\bibinfo {title} {{Lieb-Robinson Bounds
  and the Generation of Correlations and Topological Quantum Order}},\ }\href
  {https://doi.org/10.1103/PhysRevLett.97.050401} {\bibfield  {journal}
  {\bibinfo  {journal} {Phys. Rev. Lett.}\ }\textbf {\bibinfo {volume} {97}},\
  \bibinfo {pages} {050401} (\bibinfo {year} {2006})}\BibitemShut {NoStop}%
\end{thebibliography}%
\end{document}


\title{Supplemental Material for ``The Lieb-Robinson light cone for power-law interactions''}
\date{\today}
\include{authors}
\maketitle
In this Supplemental Material, we provide a rigorous proof of \cref{lem:untruncate-lattice} in the main text (\cref{sec:main-theorem-proof}) and details on the applications of the bound to connected correlators, topologically ordered states, and simulations of local observables (\cref{sec:applications}).

\tableofcontents

\section{Proof of Theorem~\ref{lem:untruncate-lattice}} \label{sec:main-theorem-proof}
In this section, we provide a rigorous proof of \cref{lem:untruncate-lattice}.
We first summarize the lemmas we use in the proof of the theorem, followed by the proofs of the lemmas in \cref{sec:recursiveproof,sec:add-long-range-proof,sec:untruncate-proof}.

For convenience, we first recall the definitions from the main text.
We consider a $d$-dimensional lattice of qubits $\Lambda$ and, acting on this lattice, a two-body power-law Hamiltonian $H(t)$ with exponent $\alpha$.
Specifically, we assume
$H(t) = \sum_{i,j\in \Lambda} h_{ij}(t)$ is a sum of two-body terms $h_{ij}$ supported on sites $i,j$ such that $\norm{h_{ij}(t)} \leq 1/\dist(i,j)^\alpha$ for all $i\neq j$,
where $\norm{\cdot}$ is the operator norm and
$\dist(i,j)$ is the distance between $i,j$.
In this paper, we assume $2d<\alpha<2d+1$.

We use $\L$ to denote the Liouvillian corresponding to the Hamiltonian $H$, i.e. $\L\oket{O} \equiv i \oket{[H,O]}$ for all operators $O$, and use $e^{\L t} \oket{O} \equiv \oket{O(t)}$ to denote the time evolved version of the operator $O$.
Similarly to the main text, we use $\P_{r}^{(i)}\oket{O}$ to denote the projection of $O$ onto sites that are at least a distance $r$ from site $i$. In particular, if $i$ is the origin of the lattice, we may also drop the superscript $i$ and simply write $\P_r$ for brevity.

Given a unit-norm operator $O$ initially supported at the origin, $\P_r e^{\L t}\oket{O}$ provides the fraction of the time-evolved version of the operator $O$ that is supported at least a distance $r$ from the origin at time $t$.
The identity~\cite{chenFiniteSpeedQuantum2019}
\begin{align}
    \frac12 \leq \frac{\norm{\P_r e^{\L t} \oket O}}{\sup_{A} \norm{\comm{A,e^{\L t} O}}} \leq 2,
\end{align}
where the supremum is taken over all unit-norm operators $A$ supported at least a distance $r$ from $O$, establishes the equivalence between the projector and the unequal-time commutator commonly used in the Lieb-Robinson literature.
\begin{theorem}\label{lem:untruncate-lattice}
	For any $\alpha\in (2d,2d+1)$ and $\epsilon\in \left(0,\frac{(\alpha - 2 d)^2}{(\alpha - 2 d)^2 + \alpha - d}\right)$, there exist constants $c,C_1,C_2\geq 0$ such that
	\begin{align}
		\norm{\P_r e^{\L t}\oket{O}} \leq
		C_1 \left(\frac{t}{r^{\alpha-2d-\epsilon}}\right)^{\frac{\alpha-d}{\alpha-2d}-\frac{\epsilon}{2}}
		+C_2\frac{t}{r^{\alpha-d}}
	\end{align}
	holds for all $t \leq c r^{\alpha-2d-\epsilon}$.
\end{theorem}

Our strategy is to divide the terms of the Hamiltonian by their interaction range and prove a Lieb-Robinson-like bound recursively for each range.
Specifically, let $\ell_0 = 0$ and $\ell_k \equiv L^k$ for $k = 1,2,\dots,n$, where $L>1$ to be chosen later,
\begin{align}
	n = \left\lfloor{\frac{1}{\log L}\log\left[r \left(\frac{t}{r^{\alpha-2d}}\right)^\eta\right]}\right\rfloor,\label{eq:n-value}
\end{align}
and $\eta \in (0,\frac{1}{\alpha-d})$ is an arbitrary small constant.
For our convenience, we set $\ell_{n+1} = r_*$, where $r_*$ is the diameter of the lattice.
We then divide the Hamiltonian into $H = \sum_{k=1}^{n+1}V_k$, where $V_k = \sum_{i,j:\ell_{k-1}<\dist(i,j)\leq \ell_k}h_{ij}$ consists of terms $h_{ij}$ such that the distance between $i,j$ is between $\ell_{k-1}$ and $\ell_k$.
We also use $H_k = \sum_{j=1}^{k} V_k$ to denote the sum of interactions whose lengths are at most $\ell_k$ and $\L_k = i [H_k,\cdot]$ are the corresponding Liouvillians.
Note that $H_{n+1} = H$ contains every interaction of the Hamiltonian.

We start with a standard Lieb-Robinson bound for $H_1$~\cite{HK,chenOperatorGrowthBounds2019}, i.e.
\begin{align}
	\norm{\mathbb P_r e^{\L_1 t}\oket{O}} \leq  \exp\left[\frac{v_1 t- r}{{\ell_1}}\right], \label{eq:bound-for-ell_1}
\end{align}
where $v_1 = 4e\tau \ell_1$ is proportional to $\ell_1$ and $\tau = \max_{i} \sum_{j\in \Lambda, j\neq i} 1/\dist(i,j)^\alpha$ is a constant for all $\alpha>d$, and recursively prove bounds for $H_2,H_3,\dots,H_n$ using the following lemma:

\begin{lemma}\label{lem:recusive}
	Suppose for ${\ell_k}\geq 1$, we have
	\begin{align}
		\norm{\mathbb P_r e^{\L_k t}\oket{O}} \leq  \exp\left[\frac{v_k t- r}{{\ell_k}}\right], \label{eq:bound-for-k}
	\end{align}
	for some unit-norm operator $O$ supported at the origin.
	Then for ${\ell_{k+1}} > {\ell_{k}}$, we have
	\begin{align}
		\norm{\mathbb P_r e^{\L_{k+1} t}\oket{O}} \leq \exp\left[\frac{v_{k+1} t - r}{{\ell_{k+1}}}\right]. \label{eq:bound-for-k+1}
	\end{align}
	where
	\begin{align}
		v_{k+1} = \xi\log(r_*) v_k + \nu \lambda \frac{\ell_{k+1}^{2d+1}}{\ell_k^\alpha}
	\end{align}
	and $\xi,\nu,\lambda$ are constants that may depend only on $d$.
\end{lemma}

Note that each of the bounds in the series has a logarithmic dependence on the diameter $r_*$ of the lattice.
We later show that this dependence on $r_*$ can be replaced by a similar logarithmic dependence on $r$, leading to a logarithmic correction in the light cone.
After applying \cref{lem:recusive} $n-1$ times, we arrive at a bound for the evolution under $H_n$:
\begin{align}
	\norm{\mathbb P_r e^{\L_n t}\oket{O}} \leq  \exp\left[\frac{v_n t- r}{{\ell_n}}\right],\label{eq:bound-n}
\end{align}
where
\begin{align}
	v_n = x^{n-1}(v_1-L^{2d+1}\nu\lambda) +x^{n-1}L^{2d+1}\nu\lambda\left[1+\frac{L^{2d+1-\alpha}}{x}+\dots \left(\frac{L^{2d+1-\alpha}}{x}\right)^{n-1}\right]
\end{align}
and $x \equiv \xi \log r_*$.
We now choose $L = x^{1/(2d+1-\alpha)}$ so that
\begin{align}
	v_n = x^{n-1} [v_1 + (n-1)L^{2d+1}\nu\lambda ].\label{eq:vn-def}
\end{align}

At this point, we have a bound for the evolution under $H_n$, which contains most terms of the Hamiltonian except for those with range larger than $\ell_n$.
With the value of $n$ in \cref{eq:n-value}, we eventually show that
the bound \cref{eq:bound-n} has the desired light cone $t \gtrsim r/v_n \sim r^{\alpha-2d}$.

Next, we add the remaining long-range interactions in $H - H_n$, i.e. those with range larger than $\ell_n$, to the bound.
The result is the following lemma, which we prove in \cref{sec:add-long-range-proof}.
\begin{lemma}\label{lem:bound-finite-lattice}
	Given any $\epsilon>0$, there exist constants $C,c,\kappa,\delta$ such that
	\begin{align}
		\norm{\P_r e^{\L t}\oket{O}} \leq C \log^\kappa r_* \left(\frac{t}{r^{\alpha-2d-\epsilon}}\right)^{\frac{\alpha-d}{\alpha-2d}-\epsilon}
	\end{align}
	holds for all $t \leq c r^{\alpha-2d-\epsilon}/\log^\delta r_*$.
\end{lemma}

The bound at this point still has an undesirable feature: it depends on the size of the lattice $r_*$.
Finally, we show in \cref{sec:untruncate-proof} that we can remove this dependence on $r_*$ at the cost of adding additional terms to the bound.
The result is \cref{lem:untruncate-lattice} presented in the main text.

\subsection{Proof of Lemma~\ref{lem:recusive}}
\label{sec:recursiveproof}
In this section, we prove \cref{lem:recusive}.
\begin{proof}
For simplicity, let $V \equiv V_{k+1} = H_{k+1} - H_{k}$ in this section.
We shall move into the interaction picture of $H_k$ and write the time evolution under $H_{k+1}$ as a product
\begin{align}
	\mathcal T \exp\left({-i\int_0^t ds\ H_{k+1}(t) }\right) = \mathcal T \exp\left({-i\int_0^t ds\ H_{k}(t) }\right).\mathcal T \exp\left({-i\int_0^t ds\ e^{\L_k s} V }\right)
\end{align}
of an evolution under $H_k$, for which \cref{eq:bound-for-k} applies, and an evolution $e^{\L_I t}$ under the $V_I(t) = e^{\L_k t} V$.

We decompose every term $h_{ij}$ in $V$ into a sum of products of two single-site operators $u_i^{(\mu)}$:
\begin{align}
	h_{ij} = \sum_{\mu}J_{ij}^{(\mu)} u_{i}^{(\mu)} u_{j}^{(\mu)},
\end{align}
where $u_i^{(\mu)}$ have unit norms, $J_{ij}^{(\mu)}$ are nonnegative, and $\sum_{\mu} J_{ij}^{(\mu)}\leq 1/\dist(i,j)^\alpha$.
In doing so, we can reduce the evolution of $h_{ij}$ into the evolutions of single-site operators $u_{i}^{(\mu)}$:
\begin{align}
	e^{\L_k t} h_{ij}
	&= \sum_{\mu} J_{ij}^{(\mu)}\left[e^{\L_k t} u_{i}^{(\mu)}\right]\left[ e^{\L_k t} u_{j}^{(\mu)}\right].
\end{align}

We then pick a parameter $R\geq \ell_k$ and divide the lattice around $i$ into shells of width $R$.
Specifically, let $\B_{r}^{(i)}$ denote the ball of radius $r$ centered on $i$.
Let $\S_{r}^{(i)} = \B_{r}^{(i)}\setminus \B_{r-R}^{(i)}$ denote the shell of inner radius $r-R$ and outer radius $r$ centered on $i$.
For each $\mu$, we have
\begin{align}
	  e^{\L_{k} t} u_{i}^{(\mu)}
	  &= \left[(\mathbb I - \P^{(i)}_{R}) + (\P^{(i)}_{ R} - \P^{(i)}_{ 2R})+ (\P^{(i)}_{ 2R} - \P^{(i)}_{3R}) +\dots \right] e^{\L_k t} u_{i}^{(\mu)}
 	  \equiv \sum_{q=0}^{\infty}u_{i,q}^{(\mu)}(t),
 \end{align}
 where the distance in the subscript of the projectors is with respect to $i$ and $u_{i,q}^{(\mu)}$ is supported on $\B_{(q+1)R}^{(i)}$ for $q = 0,1,2,\dots$.

Using \cref{eq:bound-for-k} and the triangle inequality, we can show that
\begin{align}
	\norm{u_{i,q}^{(\mu)}(t)}
	\leq \norm{\P^{(i)}_{qR} e^{\L_kt} u_i^{(\mu)}} + \norm{\P^{(i)}_{(q+1)R} e^{\L_kt} u_i^{(\mu)}}
	&\leq \exp\left(\frac{v_k t  - qR}{\ell_k}\right) + \exp\left(\frac{v_k t - (q+1)R}{\ell_k}\right).
\end{align}
Choosing $R \geq v_k t$ and $ R \geq (1+\epsilon)\ell_k$ for some positive constant $\epsilon$, we have
\begin{align}
	\norm{u_{i,q}^{(\mu)}(t)}
	\leq e^{\frac{-(q-1)R}{\ell_k}} + e^{\frac{-qR}{\ell_k}}
	\leq e^{-(q-1)(1+\epsilon)} + e^{-q(1+\epsilon)}
	\leq (1+e^{1+\epsilon})e^{-q (1+\epsilon)} \label{eq:bound-one-dumbbell}
\end{align}
for all $q = 0,1,2,\dots$.
By combining the two legs of $h_{ij}$ together, we arrive at a decomposition $e^{\L_k t} h_{ij} = \sum_{p,q} w_{i,p;j,q}(t)$,
where $w_{i,p;j,q}(t) = \sum_{\mu} J_{ij}^{(\mu)}u_{i,p}^{(\mu)}(t)u_{j,q}^{(\mu)}(t)$ and
\begin{align}
	\norm{w_{i,p;j,q}(t)} \leq \frac{(1+e^{1+\epsilon})^2}{\dist(i,j)^\alpha} e^{-(p+q)(1+\epsilon)}.
\end{align}

Next, we divide the lattice into complementary hypercubes of length $R$. We shall prove that $V_I(t)$ actually consists of exponentially decaying interactions between hypercubes.
We shall index the hypercubes by their centers, i.e. $\C_x$ denotes the hypercube center at $x$.
Given $x,y$ as the centers of two hypercubes,
\begin{align}
	\tilde h_{xy}(t) \equiv \sum_{\substack{i,j,p,q\\
	\B_{(p+1)R}^{(i)}\cap \C_x \neq \varnothing\\
	\B_{(q+1)R}^{(j)}\cap \C_y \neq \varnothing\\
	}} w_{i,p;j,q}(t)
\end{align}
defines the effective interaction between the cubes $\C_x$ and $\C_y$.
Note that $\sum_{x,y}\tilde h_{xy}\neq V_I$ because some $w_{i,p;j,q}$ might be double counted.
The conditions $\B_{(p+1)R}^{(i)}\cap \C_x \neq \varnothing$ and $
	\B_{(q+1)R}^{(j)}\cap C_y \neq \varnothing$ ensure that we account for all terms $w_{i,p;j,q}(t)$ whose support might overlap with the cubes $\C_x,\C_y$ (\cref{fig:eff-int}).
These conditions, together with $\dist(i,j)\leq \ell_{k+1}$, can be relaxed to
\begin{enumerate}
	\item $\dist(i,x) \leq (p+1)R + R \frac{\sqrt d}{2},$
	\item $\dist(j,y) \leq (q+1)R + R \frac{\sqrt d}{2},$ and
	\item $\dist(x,y) \leq (p+1)R + R\frac{\sqrt{d}}{2} + \ell_{k+1}+ (q+1)R +R \frac{\sqrt d}{2} ,$
\end{enumerate}
where $(p+1)R$ and $(q+1)R$ are the radii of the balls around $i$ and $j$, $R \sqrt{d}/2$ is the maximum distance between the center and the corner of a hypercube, and the middle term $\ell_{k+1}$ comes from the maximum distance between $i$ and $j$.

\begin{figure}
\includegraphics[width=0.45\textwidth]{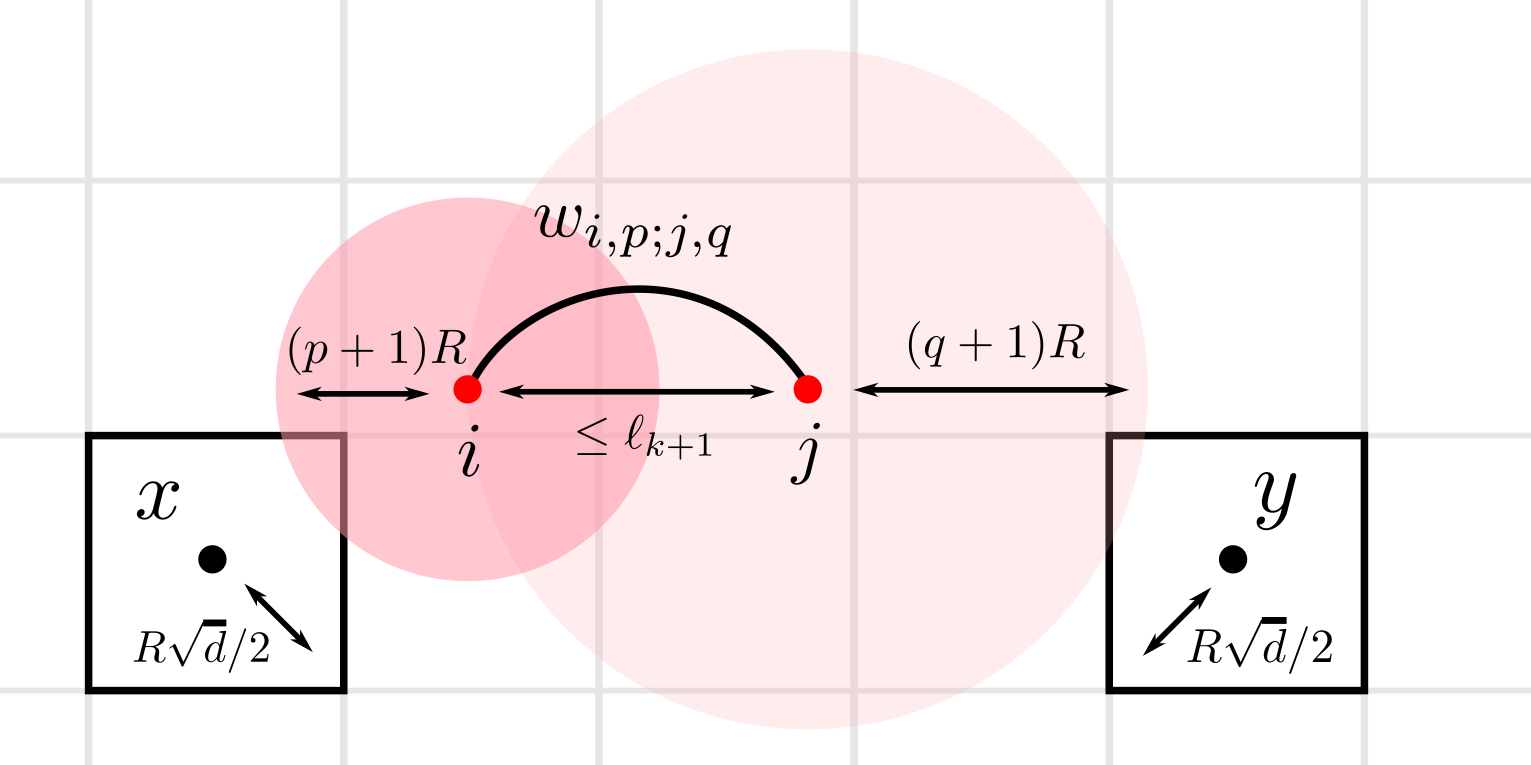}
\caption{The effective interaction between two hypercubes $\C_x$ and $\C_y$ comes from the terms $w_{i,p;j,q}$ whose support (the shaded area) overlaps with the cubes. }
\label{fig:eff-int}
\end{figure}

We bound the norm of $\tilde h_{xy}(t)$ using the triangle inequality and relax the conditions for $i,j,p,q$ as we discussed above:
\begin{align}
	\norm{\tilde h_{xy}(t)} &\leq \sum_{\substack{p,q,i,j\\(1),(2),(3)}}\norm{w_{i,p;j,q}(t)}
	\leq  \sum_{\substack{p,q,i,j\\(1),(2),(3)}} \frac{(1+e^{1+\epsilon})^2}{\dist(i,j)^\alpha} e^{-(p+q)(1+\epsilon)},
\end{align}
where the subscript $(1), (2), (3)$ of the sum refers to the three conditions above, respectively.
Since $\dist(i,j)\geq \ell_k$, we can simplify the bound and carry out the sums over $i,j$:
\begin{align}
	\norm{\tilde h_{ij}(t)} & \leq \frac{(1+e^{1+\epsilon})^2}{\ell_k^\alpha} \sum_{\substack{p,q,i,j\\(1),(2),(3)}}  e^{-(p+q)(1+\epsilon)}\\
	&\leq \frac{(1+e^{1+\epsilon})^2}{\ell_k^\alpha} \sum_{\substack{p,q\\(3)}} 4^d
	\left(R+pR + R\frac{\sqrt d}{2}\right)^d
	\left(R+qR + R\frac{\sqrt d}{2}\right)^d e^{-(p+q)(1+\epsilon)}\\
	&= \frac{(1+e^{1+\epsilon})^2}{\ell_k^\alpha} (2R)^{2d}\sum_{\substack{p,q\\(3)}}
	\left(p + 1+ \frac{\sqrt d}{2}\right)^d
	\left(q + 1+\frac{\sqrt d}{2}\right)^d
	e^{-(p+q)(1+\epsilon)}.
\end{align}
We then use use the following identity to simplify the expression:
For every $\epsilon>0$,
\begin{align}
	x^d \leq g_\epsilon e^{\epsilon x} \label{eq:power-to-exp}
\end{align}
holds for all $x\geq 0$, where $g_\epsilon = d!/\epsilon^d$.
Therefore, we can  bound
\begin{align}
	 \left(p + 1+ \frac{\sqrt d}{2}\right)^d
	 \leq g_\epsilon e^{\epsilon+\epsilon\frac{\sqrt d}{2}}e^{\epsilon p}.
\end{align}
Substituting back to the earlier equation, we have
\begin{align}
	\norm{\tilde h_{xy}(t)}
	\leq g_\epsilon^2 e^{2\epsilon+2\epsilon\frac{\sqrt d}{2}} \frac{(1+e^{1+\epsilon})^2}{\ell_k^\alpha} (2R)^{2d} \sum_{\substack{ p,q\\(3)}} e^{-(p+q)(1+\epsilon-\epsilon)}
	\leq \frac{\tilde{ g}_\epsilon}{\ell_k^\alpha} R^{2d} \sum_{\substack{ p,q\\(3)}} e^{-(p+q)},
\end{align}
where $\tilde{g}_\epsilon$ absorbs all constants that depend only on $\epsilon$ and $d$.
Recall that condition $(3)$ is equivalent to
\begin{align}
	p+q \geq \frac{\dist(x,y)}{R} - 2-\sqrt{d} - \frac{\ell_{k+1}}{R} \equiv a.
\end{align}

We consider two cases.
For $q \geq a$, the sum over $p$ can be taken from 0 to $\infty$:
\begin{align}
	\frac{\tilde{ g}_\epsilon}{\ell_k^\alpha} R^{2d} \sum_{q\geq a}\sum_{p\geq 0} e^{-(p+q)}
	&\leq  \frac{\tilde{ g}_\epsilon}{\ell_k^\alpha} R^{2d} e^{-a+1} \sum_{q\geq 0}\sum_{p\geq 0} e^{-(p+q)}
	= \frac{e^3 \tilde g_\epsilon}{(e-1)^2\ell_k^\alpha} R^{2d} e^{-a }\nonumber\\
	&= \frac{e^3 \tilde g_\epsilon}{(e-1)^2} e^{2+\sqrt{d}} e^{\frac{\ell_{k+1}}{R}} \frac{R^{2d}}{\ell_k^\alpha} e^{-\frac{\dist(i,j)}{R} }.\label{eq:pq1}
\end{align}
For $q<a$, we sum over $p\geq a-q$:
\begin{align}
	 \frac{\tilde{ g}_\epsilon}{\ell_k^\alpha} R^{2d} \sum_{q<a}\sum_{p\geq a-q} e^{-(p+q)}
	 \leq \frac{e^2}{e-1} \frac{\tilde{ g}_\epsilon}{\ell_k^\alpha} R^{2d} g_\epsilon e^{-(1-\epsilon)a}
	 &\leq \frac{e^2 \tilde g_\epsilon g_\epsilon e^{2+\sqrt{d}}}{e-1}
	 e^{\frac{\ell_{k+1}}{R}}\frac{R^{2d}}{\ell_k^\alpha}  e^{-(1-\epsilon)\frac{\dist(i,j)}{R}}
	 ,\label{eq:pq2}
\end{align}
where we have used the identity \cref{eq:power-to-exp} again with $d\geq 1$ and $\epsilon>0$ having the same value as before.

Combining \cref{eq:pq1,eq:pq2}, we have
\begin{align}
	 \norm{\tilde h_{xy}(t)}
	 &\leq \underbrace{
		 \tilde g_\epsilon \frac{e^2}{e-1} e^{2+\sqrt{d}} \left(\frac{e}{e-1}+g_\epsilon\right) e^{\frac{\ell_{k+1}}{R}}\frac{R^{2d}}{\ell_k} }_{ = \mathcal E_0}
	 e^{-(1-\epsilon)\frac{\dist(x,y)}{R}}.
\end{align}
Note that $\frac{\dist(x,y)}{R}$ is the rescaled distance between the hypercubes $\C_x$ and $\C_y$.
Therefore, the interaction between the hypercubes decays exponentially with the rescaled distance between them.
Using the standard Lieb-Robinson bound for exponentially decaying interactions, there exists a constant $\nu$ such that
\begin{align}
	\norm{\P_{r} e^{\L_I t} \oket{O}} \leq  \Exp{\nu \mathcal E_0 t - \frac{(1-\epsilon)r}{R}}
\end{align}
for any unit-norm operator $O$ supported on a single hypercube (including operators supported on single sites.)
We now choose $R = (1-\epsilon) \ell_{k+1}$ and rewrite
\begin{align}
	\mathcal E_0 =  \underbrace{
		\tilde g_\epsilon \frac{e^2}{e-1} e^{2+\sqrt{d}} \left(\frac{e}{e-1}+g_\epsilon\right)
		e^{\frac{1}{1-\epsilon}} (1-\epsilon)^{2d}
		}_{\equiv \lambda}   \frac{\ell_{k+1}^{2d}}{\ell_k^\alpha},
\end{align}
where the constant $\lambda$ depends only on $\epsilon$ and $d$.
Plugging this expression into the earlier bound, we get
\begin{align}
	\norm{\P_{r} e^{\L_I t} \oket{O}}
	\leq  \Exp{\frac{\nu \lambda  \frac{\ell_{k+1}^{2d+1}}{\ell_k^\alpha}t-r}{\ell_{k+1}}}
	=   \Exp{\frac{\Delta v\ t-r}{\ell_{k+1}}},\label{eq:bound-int}
\end{align}
where
\begin{align}
	\Delta v \equiv \nu \lambda \frac{\ell_{k+1}^{2d+1}}{\ell_k^\alpha}.
\end{align}
Note that we assume $R = (1-\epsilon)\ell_{k+1} \geq (1+\epsilon)\ell_k$.
A constant $\epsilon$ satisfying this condition exists as long as $\ell_{k+1}>\ell_k$.

Next, we use the following lemma to ``merge'' this bound for $e^{\L_I t}$ with the bound in \cref{eq:bound-for-k} for $e^{\L_k t}$.
\begin{lemma}\label{lem:combine-two-bounds}
	Let $H_1,H_2$ be two possibly time-dependent Hamiltonians and $\L_1,\L_2$ be the corresponding Liouvillians. Suppose that for all unit-norm, single-site operators $O$ and for all times $t \leq \Delta t$ for some $\Delta t$,
	\begin{align}
		 &\norm{\mathbb P_r  e^{\L_1 t} \oket{O}} \leq c_1 r^{\xi_1}e^{\frac{v_1t -r}{\ell_1}},\label{eq:subbound1}\\
		 &\norm{\mathbb P_r  e^{\L_2 t} \oket{O}} \leq c_2 r^{\xi_2}e^{\frac{v_2t -r}{\ell_2}},\label{eq:subbound2}
	\end{align}
	for some $\ell_2 \geq \ell_1$ and $c_1,c_2\geq 1; \xi_1,\xi_2 \geq 0$ are constants.
	We have
	\begin{align}
		 &\norm{\mathbb P_r  e^{\L_2 t}e^{\L_1 t} \oket{O}} \leq 2^{d+5}c_1c_2 r^{\xi_1 + \xi_2+d+1} e^{\frac{(v_1+v_2)t -r}{\ell_2}},
	\end{align}
	for all $t\leq \Delta t$.
\end{lemma}
We prove \cref{lem:combine-two-bounds} in \cref{sec:combine-two-bounds-proof}.
Using the lemma, we obtain a bound for the evolution under $H_{k+1}$:
\begin{align}
	\norm{\P_{r} e^{\L_{k+1} t} \oket{O}}
	= \norm{\P_{r} e^{\L_{I} t} e^{\L_{k} t} \oket{O}}
	 \leq 2^{d+5} r^{d+1} e^{\frac{(v_k+\Delta v)t-r}{\ell_{k+1}}}. \label{eq:bound-small-time}
\end{align}
However, because we assume $v_k t\leq R$ in deriving \cref{eq:bound-int}, \cref{eq:bound-small-time} is only valid for small time $t \leq (1-\epsilon)\ell_{k+1}/v_k \equiv \Delta t$.
To extend the bound to all time, we use a corollary of \cref{lem:combine-two-bounds}:
\begin{corollary}\label{coro:multitime}
	Suppose we have a single-site, unit-norm operator $O$, a Hamiltonian $H$ with a corresponding Liouvillian $\L$, a constant $\Delta t$,
	and
	\begin{align}
		 \norm{\P_{r} e^{\L t} \oket{O}} \leq c_0 r^{\xi_0} e^{(vt - r )/ \ell}
	\end{align}
	holds for all $t\leq \Delta t$.
	Then, for all $t \leq 2^k \Delta t$ for any $k \in \mathbb N$, we have
	\begin{align}
		 \norm{\P_{r} e^{\L t} \oket{O}} \leq c_k  r^{\xi_k} e^{(vt - r) / \ell},\label{eq:kDeltat-bound}
	\end{align}
	where
		 $c_k = 2^{(d+5)(2^{k}-1)}c_0^{2^k},$
		 $\xi_k =  (2^k-1)(d+1) + 2^k \xi_0$
	are constants.
	 In particular,
	 \begin{align}
		 \norm{\P_{r} e^{\L t} \oket{O}} \leq   e^{\frac{\chi t}{\Delta t} + \frac{vt-r}{\ell}}\label{eq:all-time-bound},
	\end{align}
	where $\chi = 2 [\log(2^{d+5}c_0) + (d+1+\xi_0)\log r]$,
	holds for all time $t$.
\end{corollary}
We prove the corollary in \cref{sec:all-time-proof}.
Using the corollary, we can extend \cref{eq:bound-small-time} to a bound for all time:
\begin{align}
	  \norm{\P_{r} e^{\L_{k+1} t} \oket{O}}
	  \leq \Exp{\frac{\chi t}{\Delta t}+\frac{(v_k+\Delta v)t-r}{\ell_{k+1}}}
	  \leq \Exp{\frac{\chi_* v_k t}{(1-\epsilon)\ell_{k+1}}+\frac{(v_k+\Delta v)t-r}{\ell_{k+1}}}
	  =\Exp{\frac{v_{k+1}t-r}{\ell_{k+1}}},
\end{align}
where we have upper bounded $\chi$ by $\chi_* = 4(d+5)\log 2+4(d+1)\log r_* $, $r_*\geq r$ is the diameter of the lattice, and
\begin{align}
	v_{k+1} &= \left(\frac{\chi_*}{1-\epsilon}+1\right) v_k + \nu \lambda \frac{\ell_{k+1}^{2d+1}}{\ell_k^\alpha}
	\leq 4(4d+13)\log(r_*) v_k + \nu \lambda \frac{\ell_{k+1}^{2d+1}}{\ell_k^\alpha}.
\end{align}
Here, we have assumed that $r_* \geq 2$ and $\epsilon\leq 1/2$ so that $1/(1-\epsilon)\leq 2$, $\chi_* \leq 4(2d+6)\log r_*$, and $1 \leq 4\log r_*$.
Therefore, \cref{lem:recusive} holds with $\xi = 4(4d+13)$.
\end{proof}
\subsubsection{Proof of \cref{lem:combine-two-bounds}}\label{sec:combine-two-bounds-proof}
In this section, we prove \cref{lem:combine-two-bounds}.
\begin{proof}
	The bound is trivial for $r < vt$, where $v = v_1+v_2$.
	Therefore, we will consider $r \geq vt$ in the rest of the proof.

	The strategy is to apply \cref{eq:subbound1,eq:subbound2} consecutively.
	A technical difficulty comes from the fact that after the first evolution $e^{\L_1 t}$, the operator has spread to more than one site.
	Therefore, we cannot directly apply \cref{eq:subbound2}, which assumes that the operator is single-site.
	Instead, we need to use \cite[Lemma 4]{tranHierarchyLinearLight2020a} to extend the bound for sing-site operators to multi-site operators.
	In particular, given the assumed bound \cref{eq:subbound2} and an unit-norm operator $O_X$ supported on a ball $X$ of radius $x\leq r$, we have
	\begin{align}
		\norm{\P_{r} e^{\L_2 t} \oket{O_X}} \leq \frac{9}{2} \abs{X} c_2 r^{\xi_2} e^{(v_2t-r+x)/\ell_2}.
	\end{align}

	With that in mind, we divide the lattice into:
	\begin{enumerate}
	 	\item A ball of radius $v_1 t$ around the origin,
	 	\item Shells of inner radius $v_1t + (q-1)\ell_1$ and outer radius $v_1t+q\ell_1$ for $q = 1,\dots, \frac{r-v_1t}{\ell_1} $,
	 	\item The rest of the lattice, i.e. sites at least a distance $r$ from the origin.
	 \end{enumerate}
	 We then project $e^{\L_1 t} \oket O$ into these regions:
	 \begin{align}
	 	e^{\L_1 t} \oket O &= \bigg[(\mathbb I - \P_{v_1t}) + \sum_{q=1}^{(r-v_1t)/{\ell_1}} (\P_{v_1t + (q-1){\ell_1}} - \P_{v_1t + q{\ell_1}}) + \P_{r}\bigg] e^{\L_1 t} \oket O\\
	 	& \equiv \oket{O_0} + \sum_{q = 1}^{q_*}  \oket{O_q} + \oket{O_{*}},
	 \end{align}
		where $q_* = (r-v_1t)/{\ell_1}$.
	 We then apply the other evolution, i.e. $e^{\L_2 t}$, on each term of the above expansion.

	First, we consider $\oket{O_0}$, which has norm at most three and is supported on at most $(2v_1t)^d = (2v_1t)^d \leq (2r)^d$ sites that are at least a distance $r-v_1t$ from the outside.
	Using the assumed bound, we have
	\begin{align}
		\norm{\P_{r} e^{\L_2 t}\oket{O_0}}
		&\leq \frac{27}{2} (2r)^d c_2  r^{\xi_2} e^{(v_2t - r+v_1t)/{\ell_2}}
		= \frac{27}{2} 2^d c_2 r^{\xi_2+d} e^{(vt-r)/{\ell_2}}.\label{eq:mid1}
	\end{align}

	Next, we consider $\oket{O_*}$.
	Because $\norm{O_*} \leq c_1 r^{\xi_1} e^{(v_1t-r)/{\ell_1}}$,
	\begin{align}
		 \norm{\P_{r} e^{\L_2 t}\oket{O_*}}
		 \leq 2\norm {O_*}
		 \leq 2c_1 r^{\xi_1} e^{(v_1t-r)/{\ell_1}}
		 \leq 2c_1c_2 r^{\xi_1+\xi_2} e^{(vt-r)/{\ell_1}}
		 \leq 2c_1c_2 r^{\xi_1+\xi_2} e^{(vt-r)/{\ell_2}}. \label{eq:mid2}
	\end{align}

	Finally, we consider $\oket{O_q}$.
	Note that $O_q$ is supported on a ball of volume at most $2^d(v_1t+q{\ell_1})^{d}\leq (2r)^d$, $\norm{O_q}\leq (1+e)c_1
	r^{\xi_1} e^{-q}$ and the distance between $O_q$ and $\P_{r}$ is $r - v_1t - q{\ell_1}\leq r$.
	Therefore, we have
	\begin{align}
		\sum_{q} \norm{\P_{r} e^{\L_2 t}\oket{O_q}} &\leq\sum_{q} \frac92(2r)^d\norm{O_q} c_2 r^{\xi_2} e^{(v_2t - (r-v_1t - q{\ell_1}))/{\ell_2}}\\
		&\leq\sum_{q} \frac92(2r)^d (1+e)c_1 r^{\xi_1} e^{-q} c_2 r^{\xi_2} e^{(vt - r)/{\ell_2}} e^{q\frac{\ell_1}{\ell_2}}\\
		&\leq \sum_{q} 17\times 2^d c_1c_2 r^{\xi_1+\xi_2+d} e^{(vt - r)/{\ell_2}}\\
		&\leq 17\times 2^d c_1c_2 r^{\xi_1+\xi_2+d+1} e^{(vt - r)/{\ell_2}}, \label{eq:mid3}
	\end{align}
	where we have used $\ell_1 \leq \ell_2$ and the fact that there are at most $\frac{r-v_1t}{\ell_1}\leq r$ different $q$.

	Combining \cref{eq:mid1,eq:mid2,eq:mid3} with $c_1,c_2\geq 1$, $d\geq 1$ and $\xi_1,\xi_2\geq 0$, we have
	\begin{align}
		\norm{\P_r e^{\L_2 t}e^{\L_1 t}\oket{O}} \leq  2^{d+5} c_1c_2 r^{\xi_1+\xi_2+d+1} e^{(vt-r)/\ell_2},
	\end{align}
	with $v = v_1+v_2$.
	Therefore, the lemma follows.
\end{proof}
\subsubsection{Proof of \cref{coro:multitime}}\label{sec:all-time-proof}
In this section, we prove \cref{coro:multitime}, an application of \cref{lem:combine-two-bounds} which extends the validity of a bound from $t\leq \Delta t$ to arbitrary time.

\begin{proof}
	The lemma clearly holds for $k = 0$. So we will prove it by induction.
	Suppose \cref{eq:kDeltat-bound} holds for some $k\in \mathbb N$.
	We will prove that it holds for $k+1$.

	The strategy is to apply the assumed bound for $k$ [\cref{eq:kDeltat-bound}] twice:
	\begin{align}
		\norm{\P_{r} e^{\L t} \oket{O}} =  \norm{\P_{r} e^{\L t/2}e^{\L t/2} \oket{O}},
	\end{align}
	where the evolutions under $e^{\L t/2}$ can be bounded by the assumed bound because $t/2 \leq 2^{k}\Delta t$.
	We then use \cref{lem:combine-two-bounds} to merge the two identical bounds with $v_1 = v_2 \rightarrow v/2$, $\ell_1 = \ell_2 \rightarrow \ell$, $c_1 = c_2 \rightarrow c_k$, $\xi_1 = \xi_2 \rightarrow \xi_k$:
	\begin{align}
		\norm{\P_r e^{\L t}\oket{O}} \leq  2^{d+5} c_k^2 r^{2\xi_k+d+1} e^{(vt-r)/\ell}.
	\end{align}
	We choose
	\begin{align}
		&c_{k+1} = 2^{d+5}c_k^2 &&\Rightarrow c_k = 2^{(d+5)(2^{k}-1)}c_0^{2^k} \\
		&\xi_{k+1} = 2\xi_k + d+1 &&\Rightarrow \xi_k = (2^k-1)(d+1) + 2^k \xi_0.
	\end{align}
	Therefore, by induction, \cref{eq:kDeltat-bound} holds for $k+1$.

	Next, to prove \cref{eq:all-time-bound}, we choose $k = \ceil{\log_2(t/\Delta t)}$ so that $t\leq 2^k \Delta t$.
	We also have $2^k \leq \frac{2t}{\Delta t}$.
	Therefore, $c_k \leq (2^{d+5}c_0)^{\frac{2t}{\Delta t}}$, $\xi_k \leq (d+1+\xi_0)\frac{2t}{\Delta t}$.
	Plugging them into \cref{eq:kDeltat-bound}, we have
	\begin{align}
		\norm{\P_{r} e^{\P_{r} t}\oket{O}} \leq  (2^{d+5}c_0)^{\frac{2t}{\Delta t}} r^{(d+1+\xi_0)\frac{2t}{\Delta t}} e^{\frac{vt-r}{\ell}} = e^{\chi \frac{t}{\Delta t} + \frac{vt - r}{\ell}},
	\end{align}
	with $\chi = 2 [\log(2^{d+5}c_0) + (d+1+\xi_0)\log r]$.
\end{proof}

\subsection{Proof of Lemma~\ref{lem:bound-finite-lattice}}\label{sec:add-long-range-proof}
In this section, we prove \cref{lem:bound-finite-lattice}.
\begin{proof}
First, we need the following lemma, which uses an existing bound to prove a tighter bound.
We will use the lemma recursively to prove the nearly optimal bound in \cref{lem:bound-finite-lattice}.
\begin{lemma}\label{lem:add-long-range-recursive}
Let $\eta \in (0,\frac{1}{\alpha-d})$ be an arbitrary constant and
\begin{align}
	\delta = \frac{2d+1}{(2d+1-\alpha)(1+\eta(2d+1-\alpha))}
\end{align}
be another constant.
Suppose there exist constants $\gamma, C, c\geq 0$, $\kappa\geq\delta$, and $\beta > d$ such that
\begin{align}
	\norm{\P_r e^{\L t}\oket{O}} \leq C \log^\kappa r_* \frac{t^{\gamma}}{r^{\beta}},\label{eq:recursive-assumption}
\end{align}
holds for all $t^{\gamma} \leq c r^{\beta}/\log^\delta r_*$.
Then, there exist constants $C', c'>0,$ and  $\kappa'>\delta$ such that
\begin{align}
	\norm{\P_r e^{\L t}\oket{O}} \leq C' \log^{\kappa'} r_* \frac{t^{\gamma'}}{r^{\beta'}},
\end{align}
holds for all $t^{\gamma'}\leq c' r^{\beta'}/\log^\delta r_*$,
where
\begin{align}
	&\kappa' = \max\left\{\kappa-\frac{\delta(\beta-d)}{\beta}+\frac{\alpha-d}{2d+1-\alpha},\delta\right\},\\
	&\gamma' = \gamma d/\beta+1-\eta(\alpha-d),\\
	&\beta'={\alpha-d-\eta(\alpha-2d)(\alpha-d)}>d.
\end{align}
\end{lemma}

\begin{proof}
Let $V = H - H_n$ to be the sum over interactions of range more than $\ell_n$.
We have~\cite{chenOperatorGrowthBounds2019}
\begin{align}
	e^{\L t} \oket{O}
	= e^{\L_n t} \oket{O}
	+\sum_{h_{ij}}\int_0^t ds\ e^{\L (t-s)} \L_{h_{ij}} e^{\L_n s}\oket{O},
\end{align}
where the sum is over all $h_{ij}$ in $V$.
The first term is the evolution under $H_n$, which we can bound using \cref{eq:bound-n}.
Our task is to bound the second term.

Without loss of generality, we assume $i\leq j$.
Because $\L_{h_{ij}}\oket{O}$ only acts nontrivially on the part of $O$ supported at least a distance $\dist(i,0)$ from the origin, we can insert $\P_{\dist(i,0)}$ in the middle of the intergrand and use the triangle inequality:
\begin{align}
	\norm{\P_{r} e^{\L t} \oket{O}}
	&\leq \norm{\P_{r} e^{\L_n t} \oket{O}}
	+\norm{\P_r\sum_{h_{ij}}\int_0^t ds e^{\L (t-s)} \L_{h_{ij}}\P_{\dist(i,0)} e^{\L_n s}\oket{O}}\\
	&\leq \norm{\P_{r} e^{\L_n t} \oket{O}}
	+4\sum_{h_{ij}}\int_0^t ds \norm{h_{ij}}\norm{\P_{\dist(i,0)} e^{\L_n s}\oket{O}}.
\end{align}

Because $\norm{h_{ij}}\leq 1/\dist(i,j)^\alpha$ and $\dist(i,j)>\ell_n$, there exist a constant $K_1$ such that
	 $\sum_{j:\dist(i,j)>\ell_n}\norm{h_{ij}} \leq {K_1}/{\ell_n^{\alpha-d}}$
for all $i\in \Lambda$.
Therefore, we have
\begin{align}
	\norm{\P_{r} e^{\L t} \oket{O}}
	&\leq \norm{\P_{r} e^{\L_n t} \oket{O}}
	+\frac{4K_1}{\ell_n^{\alpha-d}}\int_0^t ds\sum_{i}\norm{\P_{\dist(i,0)} e^{\L_n s}\oket{O}},   \label{eq:short-and-long}
\end{align}

We then consider two cases for the sum over $i$.
If $\dist(i,0)^\beta \leq \frac{1}{c}\log ^\delta (r_*)\ s^{\gamma}$, we use a trivial bound on the projection:
\begin{align}
	\sum_{i:\dist(i,0)^\beta\leq \frac{1}{c}\log ^\delta (r_*)\ s^{\gamma}}\norm{\P_{\dist(i,0)} e^{\L_n s}\oket{O}}
	\leq (2 c^{-1/\beta} (\log r_*)^{\delta/\beta}s^{\gamma/\beta})^d \times 2 = 2^{d+1}c^{-d/\beta}(\log r_*)^{\delta d/\beta}s^{\gamma d/\beta}.\label{eq:addnpart1}
\end{align}
Otherwise, if $\dist(i,0)^\beta> \frac{1}{c}\log ^\delta (r_*)\ s^{\gamma}$, we apply \cref{eq:recursive-assumption}:
\begin{align}
	\sum_{i:\dist(i,0)^\beta>\frac{1}{c}\log ^\delta(r_*)\ s^{\gamma}}\norm{\P_{\dist(i,0)} e^{\L_n s}\oket{O}}
	&\leq C\log^\kappa r_* \sum_{i: \dist(i,0)^\beta> \frac{1}{c}\log ^\delta(r_*)\ s^{\gamma}} \frac{s^{\gamma}}{\dist(i,0)^{\beta}}\\
	&\leq CK_2\log^\kappa r_*  \frac{s^{\gamma}}{[c^{-1/\beta}(\log r_*)^{\delta/\beta}s^{\gamma/\beta}]^{\beta-d}}\\
	&\leq CK_2c^{\frac{\beta-d}{\beta}}(\log r_*)^{\kappa-\frac{\delta(\beta-d)}{\beta}}\  s^{\gamma d/\beta},\label{eq:addnpart2}
\end{align}
where $K_2$ is a constant such that
\begin{align}
	\sum_{i: \dist(i,0)>  a }\frac{1}{\dist(i,0)^{\beta}} \leq  \frac{K_2}{a^{\beta-d}},
\end{align}
for all $a>0$. Such a constant $K_2$ exists because $\beta>d$ by assumption.

Combining \cref{eq:addnpart1,eq:addnpart2} and accounting for $\kappa\geq \delta$, we can upper bound
\begin{align}
	\frac{4K_1}{\ell_n^{\alpha-d}}\int_0^t ds\sum_{i}\norm{\P_{\dist(i,0)} e^{\L_n s}\oket{O}}
	&\leq K (\log r_*)^{\kappa-\frac{\delta(\beta-d)}{\beta}}\frac{t^{\frac{\gamma d}{\beta}+1}}{\ell_n^{\alpha-d}},\label{eq:long-range-part}
\end{align}
where we absorb all constants into $K = 4K_1(2^{d+1} c^{-d/\beta}+CK_2/c^{\frac{\beta-d}{\beta}})\frac{\beta}{\gamma d+\beta}$.
Substituting \cref{eq:long-range-part} in \cref{eq:short-and-long}, we have a bound for the evolution under $H$:
\begin{align}
	\norm{\P_{r} e^{\L t} \oket{O}}
	\leq e^{\frac{v_n t - r}{\ell_n}} +  K (\log r_*)^{\kappa-\frac{\delta(\beta-d)}{\beta}}\frac{t^{\frac{\gamma d}{\beta}+1}}{\ell_n^{\alpha-d}}\label{eq:bound-before-n}
\end{align}

We now substitute the values of $v_n$ and $\ell_n$ into the bound.
Recall from \cref{eq:n-value} that
\begin{align}
	n = \left\lfloor{\frac{1}{\log L}\log\left[r \left(\frac{t}{r^{\alpha-2d}}\right)^\eta\right]}\right\rfloor,
\end{align}
where $\eta \in (0,\frac{1}{\alpha-d})$ is an arbitrary small constant.
With this choice, we can bound $\ell_n$ from both above and below:
\begin{align}
	r\geq r\left(\frac{t}{r^{\alpha-2d}}\right)^\eta \geq \ell_n = L^n \geq \frac{r}{L}\left(\frac{t}{r^{\alpha-2d}}\right)^\eta = \frac{r}{(\xi\log r_*)^{1/(2d+1-\alpha)}}\left(\frac{t}{r^{\alpha-2d}}\right)^\eta. \label{eq:ell_n}
\end{align}
With $v_1 = 4e\tau \ell_1$, $x = \xi \log r_* = L^{2d+1-\alpha}$, we also have a bound for $v_n$ from \cref{eq:vn-def}:
\begin{align}
	v_n \leq r^{2d+1-\alpha}\left(\frac{t}{r^{\alpha-2d}}\right)^{\eta(2d+1-\alpha)} \frac{4e\tau x^{1/(2d+1-\alpha)} +\nu\lambda(n-1) x^{(2d+1)/(2d+1-\alpha)}}{x}.
\end{align}
Assuming that $r_* \geq e^{e^{2d+1-\alpha}/\xi}$ so that $\log L\geq 1$ and $x \geq 1$, we have $n \leq \log r /\log L \leq \log r_* = x/\xi$.
We can then crudely upper bound
\begin{align}
	v_n
	\leq  (4 e\tau+\nu\lambda/\xi)(\xi\log r_*)^{\frac{2d+1}{2d+1-\alpha}}\ \frac{r}{t} \left(\frac{t}{r^{\alpha-2d}}\right)^{1+\eta(2d+1-\alpha)}
	&=  K_3 (\log r_*)^{\frac{2d+1}{2d+1-\alpha}}  \frac{r}{t} \left(\frac{t}{r^{\alpha-2d}}\right)^{1+\eta(2d+1-\alpha)}. \label{eq:vn}
\end{align}
where $K_3$ is a constant.
Assuming
\begin{align}
	t \leq r^{\alpha-2d}/ [2K_3 (\log r_*)^{\frac{2d+1}{2d+1-\alpha}}]^{\frac{1}{1+\eta(2d+1-\alpha)}} \label{eq:broad-cond-on-t}
\end{align}
so that $v_n t \leq r/2$, we can simplify the first term of \cref{eq:bound-before-n}:
\begin{align}
	e^{\frac{v_n t-r}{\ell_n}}\leq e^{-\frac{r}{2\ell_n}}
	\leq  \exp\left[-\frac12\left(\frac{r^{\alpha-2d}}{t}\right)^\eta\right].\label{eq:full-bound-1}
\end{align}
Similarly, the second term of \cref{eq:bound-before-n} can be simplified to
\begin{align}
	K (\log r_*)^{\kappa-\frac{\delta(\beta-d)}{\beta}} \frac{t^{\gamma d/\beta+1}}{\ell_n^{\alpha-d}}
	&\leq K (\log r_*)^{\kappa-\frac{\delta(\beta-d)}{\beta}}  t^{\gamma d/\beta+1}  \left[\frac{r}{(\xi\log r_*)^{1/(2d+1-\alpha)}}\left(\frac{t}{r^{\alpha-2d}}\right)^\eta\right]^{d-\alpha}\\
	&=K \xi^{(\alpha-d)/(2d+1-\alpha)}(\log r_*)^{\kappa-\frac{\delta(\beta-d)}{\beta}+\frac{\alpha-d}{2d+1-\alpha}} \frac{t^{\gamma d/\beta+1-\eta(\alpha-d)}}{r^{\alpha-d-\eta(\alpha-2d)(\alpha-d)}}  \\
	&=K_4 \log^{\kappa'} r_* \frac{t^{\gamma'}}{r^{\beta'}},  \label{eq:full-bound-2}
\end{align}
where $K_4, \kappa', \gamma'$ are constants. In particular, $\gamma' = \gamma d/\beta+1-\eta(\alpha-d)$ and $\beta'={\alpha-d-\eta(\alpha-2d)(\alpha-d)}>d$.
Combining \cref{eq:full-bound-1,eq:full-bound-2}, we get a bound for the evolution under $H$.

We now simplify the bound by considering $t$ such that
\begin{align}
	t^{\gamma'}\leq \frac{c' r^{\beta'}}{(\log r_*)^{\delta \gamma'}}\label{eq:new-cond-on-t}
\end{align}
for some constant $c'$ and $\delta = \frac{2d+1}{(2d+1-\alpha)(1+\eta(2d+1-\alpha))}$.
Since $\beta/\gamma \leq \alpha-2d$ by assumption, we also have
\begin{align}
	\frac{\beta'}{\gamma'} &= \frac{\alpha-d-\eta(\alpha-2d)(\alpha-d)}{\frac{\gamma d}{\beta}+1 - \eta(\alpha-d)}
	\leq \frac{\alpha-d-\eta(\alpha-2d)(\alpha-d)}{\frac{d}{\alpha-2d}+1 - \eta(\alpha-d)}
	= \alpha-2d.
\end{align}
Therefore, with $c' = (2K_3)^{\frac{-\gamma'}{(1+\eta(2d+1-\alpha))}}$,
\cref{eq:new-cond-on-t} satisfies the condition in \cref{eq:broad-cond-on-t}.
In addition, for $r^{\alpha-2d}\geq t$, there exists a constant $K_5$ such that
\begin{align}
	\exp\left[-\frac12\left(\frac{r^{\alpha-2d}}{t}\right)^\eta\right]
	\leq K_5 \left(\frac{t}{r^{\alpha-2d}}\right)^{\gamma'}
	\leq K_5 \frac{t^{\gamma'}}{r^{\beta'}},\label{eq:full-bound-1b}
\end{align}
where we have again used $\gamma'(\alpha-2d)\geq \beta'$ in the last inequality.
Replacing \cref{eq:full-bound-1} by \cref{eq:full-bound-1b} and combining with \cref{eq:full-bound-2}, we arrive at a bound
\begin{align}
	\norm{\P_r e^{\L t}\oket{O}} \leq C' \log^{\kappa'} r_* \frac{t^{\gamma'}}{r^{\beta'}},\label{eq:dhaskds}
\end{align}
for all $t^{\gamma'}\leq c' r^{\beta'}/\log^\delta r_*$,
where $C' \geq K_4 + K_5$ and $c'$ are constants,
\begin{align}
	&\kappa' = \kappa-\frac{\delta(\beta-d)}{\beta}+\frac{\alpha-d}{2d+1-\alpha},\\
	&\gamma' = \gamma d/\beta+1-\eta(\alpha-d),\\
	&\beta'={\alpha-d-\eta(\alpha-2d)(\alpha-d)}>d.
\end{align}
If $\kappa' <\delta$, we simply replace $\kappa'$ by $\delta$ in \cref{eq:dhaskds}. Such replacement can only increase the bound in \cref{eq:dhaskds}.
Therefore, \cref{lem:add-long-range-recursive} follows.
\end{proof}

We now use \cref{lem:add-long-range-recursive} to prove \cref{lem:bound-finite-lattice}.
To satisfy the assumption of \cref{lem:add-long-range-recursive}, we start with the bound in Ref.~\cite{Foss-FeigG}:
There exist constants $K_6,K_7$, and $v_F$ such that
\begin{align}
	\norm{\comm{O',e^{\L t}O}} \leq K_6 \Exp{v_F t - \frac{r}{t^{(1+d)/(\alpha-2d)}}} + K_7 \frac{t^{\frac{\alpha(\alpha-d+1)}{\alpha-2d}}}{r^{\alpha}},\label{eq:mike}
\end{align}
for all single-site, unit-norm operators $O'$ supported a distance $r$ from $O$.
We consider the regime
\begin{align}
	t^{\frac{\alpha(\alpha-d+1)}{\alpha-2d}} \leq c r^{\alpha-d}/\log^\delta r_* \leq c r^{\alpha-d} \leq  c r^{\alpha},\label{eq:mike-t-regime}
\end{align}
where we choose $c = (2v_F)^{-\alpha}$ so that
\begin{align}
	v_F t \leq  \frac{r}{2t^{(1+d)/(\alpha-2d)}}.
\end{align}
Therefore, there exists a constant $K_8$ such that
\begin{align}
	K_6 \Exp{v_F t - \frac{r}{t^{(1+d)/(\alpha-2d)}}} \leq K_6 \Exp{ - \frac{r}{2t^{(1+d)/(\alpha-2d)}}}\leq K_8 \left(\frac{t^{(1+d)/(\alpha-2d)}}{r}\right)^{\frac{\alpha(\alpha-d+1)}{1+d}}
	\leq K_8 \frac{t^{\alpha(\alpha-d+1)/(\alpha-2d)}}{r^{\alpha}} \label{eq:mike-sim-1}
\end{align}
holds for all $t$ satisfying \cref{eq:mike-t-regime}.
In the last inequality, we have used $\alpha-d+1 \geq d+1$ to lower bound the exponent of $r$.
Substituting \cref{eq:mike-sim-1} into \cref{eq:mike}, we get a simplified version of the bound in Ref.~\cite{Foss-FeigG}:
\begin{align}
	\norm{\comm{O',e^{\L t}O}} \leq K_9 \frac{t^{\frac{\alpha(\alpha-d+1)}{\alpha-2d}}}{r^{\alpha}},\label{eq:mike2}
\end{align}
where $K_9 = K_7 + K_8$.
Applying Lemma 4 in Ref.~\cite{tranHierarchyLinearLight2020a}, there exists a constant  $K_{10}$ such that:
\begin{align}
	\norm{\P_r e^{\L t}\oket{O}} \leq K_{10} \frac{t^{\frac{\alpha(\alpha-d+1)}{\alpha-2d}}}{r^{\alpha-d}}
	< K_{10} \log^{\delta}r_* \frac{t^{\frac{\alpha(\alpha-d+1)}{\alpha-2d}}}{r^{\alpha-d}},\label{eq:mike3}
\end{align}
where the additional factor $-d$ in the exponent of $r$ comes from ``integrating'' over sites that are at least a distance $r$ from the origin.
\Cref{eq:mike3} satisfies the assumption of \cref{lem:add-long-range-recursive}, with $C \rightarrow K_{10}, c \rightarrow (2v_F)^{-\alpha}, \kappa \rightarrow \delta, \gamma \rightarrow \frac{\alpha(\alpha-d+1)}{\alpha-2d}, \beta \rightarrow \alpha-d$.
Therefore, by the lemma, there exist constants $C_1,c_1,\kappa_1$ such that
\begin{align}
	\norm{\P_r e^{\L_r}\oket{O}} \leq C_1 \log^{\kappa_1} r_* \frac{t^{\gamma_1}}{r^{\beta_1}}\label{eq:first-iteration}
\end{align}
holds for all $t^{\gamma_1} \leq c_1 r^{\beta_1}/\log^\delta r_*$,
where
\begin{align}
	&\gamma_1 = \frac{\alpha(\alpha-d+1)}{\alpha-2d}\frac{d}{\alpha-d}+1 - \eta (\alpha-d),\\
	&\beta_1 = \alpha-d-\eta(\alpha-2d)(\alpha-d).
\end{align}
\Cref{eq:first-iteration} again satisfies the assumption of \cref{lem:add-long-range-recursive}.
Applying the lemma again with $\gamma\rightarrow \gamma_1,\beta\rightarrow \beta_1$, we obtain
\begin{align}
	\norm{\P_r e^{\L_r}\oket{O}} \leq C_2 \log^{\kappa_2} r_* \frac{t^{\gamma_2}}{r^{\beta_2}}\label{eq:second-iteration}
\end{align}
for some constants $C_2,\kappa_2, \beta_2 = \beta_1$, and
\begin{align}
	\gamma_2 = \frac{\gamma_1 d}{\beta_1} + 1 -\eta(\alpha-d) \equiv f(\gamma_1).
\end{align}
After applying \cref{lem:add-long-range-recursive} for $m$ times, we obtain
\begin{align}
	\norm{\P_r e^{\L_r}\oket{O}} \leq C_m \log^{\kappa_m} r_* \frac{t^{\gamma_m}}{r^{\beta_m}}
	= C_m \log^{\kappa_m} r_* \left(\frac{t}{r^{\beta_m/\gamma_m}}\right)^{\gamma_m}\label{eq:mth-iteration}
\end{align}
for some constants $C_m,\kappa_m, \beta_m = \beta_1$, and
\begin{align}
	\gamma_m = f^{\circ (m-1)}(\gamma_1),
\end{align}
where $f^{\circ (m-1)}$ denotes the $(m-1)$-th composition of the function $f$.
It is straightforward to show that
\begin{align}
	&\lim_{\eta\rightarrow 0}\lim_{m\rightarrow \infty}\gamma_m = \lim_{\eta\rightarrow 0}   \frac{\alpha-d -\eta(\alpha-2d)(\alpha-d)}{\alpha-2d} = \frac{\alpha-d}{\alpha-2d},\\
	&\lim_{\eta\rightarrow 0}\lim_{m\rightarrow \infty} \frac{\beta_m}{\gamma_m} = \alpha-2d.
\end{align}
Therefore, for all $\epsilon>0$, there exist $m\geq 1,\eta\in (0,\frac{1}{\alpha-d})$ such that $\beta_m/\gamma_m \geq \alpha-2d - \epsilon$ and $\gamma_m \geq \frac{\alpha-d}{\alpha-2d}-\epsilon$.
We obtain
\begin{align}
	\norm{\P_r e^{\L_r}\oket{O}} \leq C_m \log^{\kappa_m} r_* \left(\frac{t}{r^{\alpha-2d-\epsilon}}\right)^{\frac{\alpha-d}{\alpha-2d}-\epsilon},
\end{align}
which holds for all $t \leq c_m^{1/\gamma_m} r^{\alpha-2d-\epsilon}/(\log r_*)^{\delta/\gamma_m} \leq r^{\alpha-2d-\epsilon}$.
\Cref{lem:bound-finite-lattice} thus follows.
\end{proof}

\subsection{Removing the dependence on the lattice size}\label{sec:untruncate-proof}
In this section, we use \cref{lem:bound-finite-lattice} to prove \cref{lem:untruncate-lattice} by removing the dependence on $r_*$.

\begin{proof}
Let $H_\out = \P_{r_0} H$ denote the terms of the Hamiltonian $H$ that have support outside a distance $r_0$ from the origin, $H_\ins = H - H_\out$ be the rest of the Hamiltonian, and $\L_\out$, $\L_\ins$ are the corresponding Liouvillians.
Using the triangle inequality, we have
\begin{align}
	\norm{\P_r e^{\L t}\oket{O}}
	\leq \norm{\P_r e^{\L_\ins t}\oket{O}} + \sum_{h_{ij}} \norm{\P_r \int_0^t ds\ e^{\L (t-s)} \L_{h_{ij}}e^{\L_\ins s} \oket{O}},	\label{eq:untruncate-1}
\end{align}
where the sum is taken over terms $h_{ij}$ in $H_\out$.
Without loss of generality, we assume $\dist(i,0)\leq \dist(j,0)$, which implies $\dist(j,0)\geq r_0$.
In addition, since $e^{\L_\ins s}\oket{O}$ is supported entirely within the radius $r_0$ from the origin, only terms where $\dist(i,0)\leq r_0$ contribute to the above sum.
We consider two cases: $\dist(i,0)>r_0/2$ and $\dist(i,0)\leq r_0/2$.

In the former case, we insert $\P_{\dist(i,0)}$ in the middle of the integrand and bound
\begin{align}
	&\sum_{h_{ij}:\dist(i,0)\in(\frac{r_0}2,r_0]} \norm{\P_r \int_0^t ds\ e^{\L (t-s)} \L_{h_{ij}}\P_{\dist(i,0)}e^{\L_\ins s} \oket{O}}
	\leq 4\sum_{h_{ij}:\dist(i,0)\in(\frac{r_0}2,r_0]}\norm{h_{ij}}
	\int_0^t ds \norm{\P_{\dist(i,0)} e^{\L_\ins s}\oket{O}}\nonumber\\
	&\leq 4K_1 C\log^{\kappa}(2r_0) \sum_{i:\dist(i,0)\in(\frac{r_0}2,r_0]} \int_0^t ds
	\left(\frac{t}{\dist(i,0)^{\alpha-2d-\epsilon}}\right)^{\frac{\alpha-d}{\alpha-2d}-\epsilon}
	\leq K_2 r_0^d t\left(\frac{t}{r_0^{\alpha-2d-\epsilon}}\right)^{\frac{\alpha-d}{\alpha-2d}-\epsilon}, \label{eq:untruncate-2}
\end{align}
where $K_2$ is a constant.
We have used \cref{lem:bound-finite-lattice} to bound the evolution under $e^{\L_\ins s}$, which is supported entirely within a truncated lattice of diameter $2r_0$, and used the fact that the interaction $h_{ij}$ decays as a power law with an exponent $\alpha > 2d$ to bound the sum over $j$ by a constant.
We require $t \leq c_1 r^{\alpha-2d-\epsilon}/\log^{\delta}(2r_0)$, for some constant $c_1, \delta$, to satisfy the conditions of \cref{lem:bound-finite-lattice}.

On the other hand, when $\dist(i,0)\leq r_0/2$, we have $\dist(i,j)\geq r_0/2$.
Therefore, there exists a constant $c_2$ such that $\sum_{j} \norm{h_{ij}} \leq c_2 / r_0^{\alpha-d}$ for all $i$.
We can then bound
\begin{align}
	\sum_{h_{ij}:\dist(i,0)\leq \frac{r_0}2} \norm{\P_r \int_0^t ds\ e^{\L (t-s)} \L_{h_{ij}}e^{\L_\ins s} \oket{O}}
	\leq 4\sum_{i:\dist(i,0)\leq\frac{r_0}2}\frac{c_2}{r_0^{\alpha-d}} \int_0^t ds
	&\leq K_3\frac{t}{r_0^{\alpha-2d}}, \label{eq:untruncate-3}
\end{align}
for some constant $K_3$.

Using \cref{lem:bound-finite-lattice} on the first term of \cref{eq:untruncate-1} and combining with \cref{eq:untruncate-2,eq:untruncate-3}, we have:
\begin{align}
	\norm{\P_r e^{\L t}\oket{O}} \leq
	C \log^\kappa(2r_0) \left(\frac{t}{r^{\alpha-2d-\epsilon}}\right)^{\frac{\alpha-d}{\alpha-2d}-\epsilon}
	+K_2 r_0^d t\left(\frac{t}{r_0^{\alpha-2d-\epsilon}}\right)^{\frac{\alpha-d}{\alpha-2d}-\epsilon}
	+ K_3\frac{t}{r_0^{\alpha-2d}}.
\end{align}
We choose $r_0 = r^{\xi}$, where
\begin{align}
	\xi = \left(1 + \frac{\alpha-d}{\alpha-2d}-\epsilon\right)\frac{\alpha-2d-\epsilon}{\alpha-2d-\epsilon\frac{(\alpha-2d)^2+\alpha-d}{\alpha-2d}+\epsilon^2} \geq \frac{\alpha-d}{\alpha-2d}, \label{eq:untruncr0}
\end{align}
and we require $\epsilon \leq (\alpha - 2 d)^2/[(\alpha - 2 d)^2 + \alpha - d]$ so that the lower bound on $\xi$ holds.
Under this choice,
\begin{align}
	K_2 r_0^d t\left(\frac{t}{r_0^{\alpha-2d-\epsilon}}\right)^{\frac{\alpha-d}{\alpha-2d}-\epsilon}
	= K_2  \left(\frac{t}{r^{\alpha-2d-\epsilon}}\right)^{1+\frac{\alpha-d}{\alpha-2d}-\epsilon}
	\leq K_4  \left(\frac{t}{r^{\alpha-2d-\epsilon}}\right)^{\frac{\alpha-d}{\alpha-2d}-\epsilon}\label{eq:untrunc-bound1},
\end{align}
for all $t \leq c_1 r^{\alpha-2d-\epsilon}$, where $K_4$ is a constant.
In addition, for $\epsilon \leq (\alpha - 2 d)^2/[(\alpha - 2 d)^2 + \alpha - d]$, $\xi \geq (\alpha-d)/(\alpha-2d)$ and, therefore,
\begin{align}
	K_3 \frac{t}{r_0^{\alpha-2d}} \leq K_3 \frac{t}{r^{\alpha-d}}.\label{eq:untrunc-bound2}
\end{align}
Combining \cref{eq:untruncr0,eq:untrunc-bound1,eq:untrunc-bound2}, we have
\begin{align}
	\norm{\P_r e^{\L t}\oket{O}} \leq
	K_5 \log^\kappa(r) \left(\frac{t}{r^{\alpha-2d-\epsilon}}\right)^{\frac{\alpha-d}{\alpha-2d}-\epsilon}
	+K_6\frac{t}{r^{\alpha-d}},\label{eq:final-with-log}
\end{align}
which holds for all $t \leq c_1 r^{\alpha-2d-\epsilon}/\log^\delta(2r^\xi)$ for some constants $K_5,K_6$ independent of $t,r$.

Next, we simplify \cref{eq:final-with-log} by ``hiding'' the factor $\log^\kappa r$ inside the constant $\epsilon$.
Specifically, there exist a constant $K_7$ such that $\log^\kappa r \leq K_7 r^{\epsilon'}$, where $\epsilon' = \frac{\epsilon}{2}\left(\frac{\alpha-d}{\alpha-2d}-\frac{\epsilon}{2}\right)$, and constants $K_8,K_9$ such that
\begin{align}
	\norm{\P_r e^{\L t}\oket{O}} \leq
	K_8 \log^\kappa(r) \left(\frac{t}{r^{\alpha-2d-\epsilon/2}}\right)^{\frac{\alpha-d}{\alpha-2d}-\frac{\epsilon}{2}}
	+K_9\frac{t}{r^{\alpha-d}}
	\leq
	K_7K_8  \left(\frac{t}{r^{\alpha-2d-\epsilon}}\right)^{\frac{\alpha-d}{\alpha-2d}-\frac{\epsilon}{2}}
	+K_9\frac{t}{r^{\alpha-d}},
\end{align}
which holds for all $t \leq c_1 r^{\alpha-2d-\epsilon/2}/\log^\delta(2r^{\xi})$.
In addition, there exists a constant $K_{10}$ such that $K_{10} \log^\delta(2r^{\xi}) \leq r^{\epsilon/2}$ for all $r\geq 1$.
By requiring that $t \leq K_{10} c_1 r^{\alpha-2d-\epsilon}$, we also ensure $t \leq c_1 r^{\alpha-2d-\epsilon/2}/\log^\delta(2r^{\xi})$.
Therefore, \cref{lem:untruncate-lattice} follows with $c \rightarrow c_1 K_{10}, C_1 \rightarrow K_7K_8$, and $ C_2 \rightarrow K_9$.

\end{proof}

\section{Applications of Theorem~\ref{lem:untruncate-lattice}}\label{sec:applications}
We discussed in the main text that the tightened light cone and nearly optimal tail in \cref{lem:untruncate-lattice} improved the scaling for various applications of Lieb-Robinson bounds to problems of physical interest in the regime $2d < \alpha < 2d + 1$
Here we provide some mathematical details to justify those assertions.
We also provide a table briefly summarizing the bounds we will use to compare, where we consider each bound to take the form $\norm{[A(t), B]} \leq c t^{\gamma}/r^{\beta}$ for some constants $c$, $\gamma$ and $\beta$, where $A$ is single-site, but $B$ may generally be some large multi-site operator.
\begin{table}[ht!]
\centering
\begin{tabular}{lllllll}
\toprule
Bound & Light cone                                       & Tail & $\gamma$ & $\beta$ & $\gamma'$ & $\beta'$ \\ \hline
&&&\vspace*{-0.1in}&\\
 This work (B1) & $t \gtrsim r^{\alpha - 2d}$           & $1/r^{\alpha - d}$ & $\frac{\alpha - d}{\alpha - 2d}$ & $\alpha - d$ & $\gamma + 1$ & $\beta$   \\
Ref.~\cite{Foss-FeigG} (B2)  & $t \gtrsim r^{\frac{(\alpha-d)(\alpha - 2d)}{\alpha(\alpha -d +1)}}$ & $1/r^{\alpha - d}$  & $\frac{\alpha(\alpha - d + 1)}{\alpha - 2d}$ & $\alpha - d$ & $\gamma + 1$ & $\beta$ \\
Ref.~\cite{tranLocalityDigitalQuantum2019a} (B3) & $t \gtrsim r^{\frac{\alpha - 2d}{\alpha - d}}$ & $1/ r^{\alpha - 2d}$ & $\alpha - d$ & $\alpha - 2d$ & $\gamma$ & $\beta$ \vspace*{0.02in}\\ \botrule
\end{tabular}
\caption{Comparison of Lieb-Robinson bounds for $2d<\alpha<2d+1$. We name the bounds B1, B2, and B3 for brevity. We ignore the arbitrarily small parameter $\epsilon$ in B1 for simplicity, as it does not affect the conclusions.}
\label{tab:bound_comparison}
\end{table}

We first consider the application of the bound on the growth of connected correlators.
Consider two unit-norm, single-site observables $A$ and $B$ initially supported on sites $x$ and $y$, respectively, such that $x$ and $y$ are separated by a distance $r$. Let $\ket{\psi}$ be a product state between $\mathcal{B}_{r/2}(x)$ and $\mathcal{B}_{r/2}(y)$, where $\mathcal{B}_{r/2}(x)$ is the ball of radius $r/2$ around $x$.
The connected correlator is defined by
\begin{equation}
    C(r,t) \equiv \avg{A(t)B(t)} - \avg{A(t)}\avg{B(t)},
\end{equation}
where $\avg{\cdot} \equiv \bra{\psi}\cdot\ket{\psi}$.
Define $\tilde{A}(t) \equiv \Tr_{\mathcal{B}_{r/2}^{c}(x)}[A(t)]$ and $\tilde{B}$ similarly.
It is elementary to bound $C(r,t)$ by
\begin{equation}
    C(r,t) \leq 2\norm{A(t)-\tilde{A}(t)} + 2\norm{B(t)-\tilde{B}(t)}.
\end{equation}
That is, the connected correlator is controlled by the error in truncating $A(t)$ and $B(t)$ to within a ball of radius $r/2$ around their initial support.
A simple result from Ref.~\cite{Bravyi06} allows us to bound this error
 \begin{equation}
       \norm{A(t) - \tilde{A}(t)} \leq \int_{\mathcal{B}_{r/2}^{c}(x)} dU \norm{[U,A(t)]} \leq c \frac{t^{\gamma}}{(r/2)^{\beta}},
    \end{equation}
where $dU$ is the Haar measure on unitaries supported outside a ball of radius $r/2$ around $x$.
Thus, for a given Lieb-Robinson bound
\begin{equation}
    C(r,t) \leq 2^{\beta+2}c\frac{t^{\gamma}}{r^{\beta}}.
\end{equation}
Ignoring constants and focusing on the asymptotics with respect to $t$ and $r$, we see that
\begin{align}
    R_{12} \equiv \frac{\text{B1}}{\text{B2}} &\sim \left(\frac{t^{\alpha-d}}{t^{\alpha(\alpha - d + 1)}} \right)^{\frac{1}{\alpha - 2d}}, \\
    R_{13} \equiv \frac{\text{B1}}{\text{B3}} &\sim \left(\frac{t^\frac{1}{\alpha - 2d}}{t}\right)^{\alpha -d}r^{-d}.
\end{align}
Thus, as $t$ increases, the tighter light cone of B1 leads to significant improvement in bounding the connected correlator as compared to B2.
While B1 has a slightly worse time-dependence than B3 (as $0 < \alpha - 2d < 1$), it has a much better $r$-dependence.
And, of course, when taken together, B1 follows a tighter light cone than B3, leading to an overall more useful bound.
Thus, while B3 may strictly have a better time-dependence, B1 provides the tightest holistic bound on the growth of connected correlators.

A nearly identical calculation allows us to place stricter bounds on the time required to generate topologically ordered states from topologically trivial ones.
We define topologically ordered states as follows: consider a lattice $\Lambda$ with diameter $L$ and $O(L^{d})$ sites.
We say that a set of orthonormal states $\{\ket{\psi_{1}}, \dots, \ket{\psi_{k}}\}$ are topologically ordered if there exists a constant $ \delta$ such that
\begin{equation}
    \epsilon \equiv \sup_{O}\max_{1 \leq i, j \leq k}\{\abs{\bra{\psi_{i}}O\ket{\psi_{j}} - \bra{\psi_{j}}O\ket{\psi_{i}}, 2\bra{\psi_{i}}O\ket{\psi_{j}}}\}
\end{equation}
is bounded $\epsilon = O(L^{-\delta})$.
The supremum is taken over operators $O$ supported on a subset of the lattice with diameter $\ell' < L$, so $\epsilon$ essentially measures the ability to distinguish between states $\ket{\psi_{i}}$ using an operator $O$ supported on only a fraction of the lattice.
In contrast, we say the states are topologically trivial if $\epsilon$ is independent of $L$.
Given a set of topologically trivial states $\{\ket{\phi_{i}}\}_{i}$ and a set of topologically ordered states $\{\ket{\psi_{i}}\}_{i}$, the question is how long it takes to generate a unitary $U$ such that $U\ket{\phi_{i}} = \ket{\psi_{i}}$ for all $i$ using a power-law Hamiltonian.
Ref.~\cite{tranHierarchyLinearLight2020a} proves that this time is controlled by the time it takes $\norm{O(t) - O(t,\ell')}$ to become non-vanishing in $L$, where $O(t, \ell')$ is the truncation of the time evolution of $O$ to a radius $\ell'$.
This expression is bounded in the exact same way as the connected correlator was, and so we see the same improvement from B1 as compared to both B2 and B3.

Finally, we consider the task of simulating the evolution of a local observable under a power-law Hamiltonian $H$ using quantum simulation algorithms.
In contrast to the earlier applications, it is not sufficient to simply truncate the time-evolved observable to the light cone.
Instead, to simulate the observable, we need to construct the Hamiltonian that generates the dynamics of the observable inside the light cone.

Let $A$ be a unit-norm, single-site observable originally supported on site $x$, and consider $A(t)$ its evolution under a 2-local power-law Hamiltonian $H$.
Define $H_{r}$ to be the Hamiltonian constructed by taking terms of $H$ that are fully supported within $\mathcal{B}_{r}(x)$, and let $\tilde{A}(t)$ be $A(0)$ evolved under $H_{r}$ (note that this is different than our previous definition of $\tilde{A}$).
The question is how large $r$ must be (i.e., how many terms of $H$ must we simulate) for $\norm{A(t)-\tilde{A}(t)}$ to have small error.
Intuitively, this observable should be constrained to lie mostly within the light cone of a Lieb-Robinson bound for $H$ as long as the tail of the bound decays sufficiently quickly, so we expect $r$ to be related to the lightcone of our bounds.
Refs.~\cite{tranHierarchyLinearLight2020a,tranLocalityDigitalQuantum2019a} make this intuition rigorous and yield
\begin{equation}
\norm{A(t) - \tilde{A}(t)} \lesssim \frac{t^{\gamma'}}{r^{\beta'}},
\end{equation}
where $\gamma'$ and $\beta'$ are listed in \cref{tab:bound_comparison}.
In particular, in order to ensure only a constant error, we must choose $r \sim t^{\gamma'/\beta'}$, which corresponds to simulating about $r^{2} \sim t^{2d\gamma'/\beta'}$ terms of the Hamiltonian.
We can compare this exponent $\phi \equiv \gamma'/\beta'$ between bounds:
\begin{align}
    \phi_{\text{B1}} - \phi_{\text{B2}} &= -\frac{(\alpha-1)(\alpha - d) + \alpha}{(\alpha -d)(\alpha - 2d)}, \\
    \phi_{\text{B1}} - \phi_{\text{B3}} &= -\frac{(\alpha-d)^{2} + d}{(\alpha-2d)(\alpha-d)}.
\end{align}
These differences are all negative for $2d < \alpha <2d+1$, meaning the current work provides the tightest bound on how many terms must be kept to get constant error when simulating the evolution of local observables in this regime.

\makeatletter
\renewcommand\@biblabel[1]{[S#1]}
\makeatother
\bibliography{my-bib}